\documentclass{article}
\usepackage[utf8]{inputenc}
\usepackage[english]{babel}
\usepackage{amsthm}
\newtheorem*{remark}{Remark}

\newtheorem{theorem}{Theorem}[subsection]

\newtheorem{lemma}[theorem]{Lemma}
\usepackage{amsmath,amssymb,subfigure}
\usepackage{graphicx}
\usepackage[pdftex,bookmarks=true,
plainpages=true,
pdfpagelabels,
hypertexnames=false,
bookmarksnumbered=false,
breaklinks=true, colorlinks=true,
urlcolor=blue, linkcolor=blue]{hyperref}
\usepackage{comment}
\usepackage{float}

\date{}
\begin{document}

\title{A subordinated CIR intensity model with application to Wrong-Way risk CVA}
\author{Cheikh Mbaye $\qquad$ Fr\'ed\'eric Vrins \footnote{Voie du Roman Pays 34, B-1348 Louvain-la-Neuve, Belgium.E-mail: \href{mailto:frederic.vrins@uclouvain.be}{frederic.vrins@uclouvain.be}.}\vspace{0.2cm}\\ Louvain Finance Center (LFIN) \& CORE \footnote{Center for Operations Research and Econometrics.}\vspace{0.2cm}\\ \textit{Universit\'e catholique de Louvain, Belgium}}

\maketitle
\font\dsrom=dsrom10 scaled 1200
\def \indic{\textrm{\dsrom{1}}}
\begin{abstract}
Credit Valuation Adjustment (CVA) pricing models need to be both flexible and tractable. The survival probability has to be known in closed form (for calibration purposes), the model should be able to fit any valid Credit Default Swap (CDS) curve, should lead to large volatilities (in line with CDS options) and finally should be able to feature significant Wrong-Way Risk (WWR) impact. The Cox-Ingersoll-Ross model (CIR) combined with independent positive jumps and deterministic shift (JCIR++) is a very good candidate : the variance (and thus covariance with exposure, i.e. WWR) can be increased with the jumps, whereas the calibration constraint is achieved via the shift. In practice however, there is a strong limit on the model parameters that can be chosen, and thus on the resulting WWR impact. This is because only non-negative shifts are allowed for consistency reasons, whereas the upwards jumps of the JCIR++ need to be compensated by a downward shift. To limit this problem, we consider the two-side jump model recently introduced by Mendoza-Arriaga \& Linetsky, built by time-changing CIR intensities. In a multivariate setup like CVA, time-changing the intensity partly kills the potential correlation with the exposure process and destroys WWR impact. Moreover, it can introduce a forward looking effect that can lead to arbitrage opportunities. In this paper, we use the time-changed CIR process in a way that the above issues are avoided. We show that the resulting process allows to introduce a large WWR effect compared to the JCIR++ model. The computation cost of the resulting Monte Carlo framework is reduced by using an adaptive control variate procedure.
\end{abstract}

\section{Introduction}
Since the 2008 crisis, regulators suggest financial institutions to pay specific attention to Counterparty Credit Risk (CCR) when valuing Over the Counter (OTC) deals. In this context, CCR
refers to the possibility that the counterparty of the transaction can default before the maturity of the contract. The CCR can be accounted for either by setting up a strong collateralisation agreement, or by charging a Credit Value Adjustment (CVA) to absorb the corresponding expected losses. One of the main challenge when pricing such adjustments is to account for the potential dependency of the exposure with counterparty's credit quality, a phenomenon commonly referred to as wrong-way risk (WWR).

In that respect, a popular framework fitting in the class of reduced-form models is to consider the default intensity to be governed by the CIR++~\cite{BA} or the JCIR++~\cite{BE} process. In essence, the intensity is modeled as a CIR or a jump diffusion CIR (JCIR) dynamics shifted in a deterministic way so as to fit a given CDS term structure. However, this model suffers from an important restriction: the resulting intensity process (including the shift) needs to be positive. Because for tractability reasons the jumps in the JCIR++ model are upwards only, increasing the jump activity (e.g. to increase the implied spread volatility) under the constraint of keeping the survival probability curve unchanged introduces a shift function that tends to be more and more negative. Because negative shifts should be ruled out for consistency reasons, this puts limits on the CIR or JCIR parameters that can be chosen. This will further limit the WWR impact in CVA applications. One way to limit the appearance of shift functions with negative values would be to allow for both upwards and downwards jumps, without affecting the tractability of the model. 

In this paper, we consider the approach of Mendoza-Arriaga and Linetsky \cite{LIN} to model the default intensity as a time-changed CIR in a CVA context. This poses several problems that need to be addressed. First, the time-change process will destroy the potential correlation between the intensity and the exposure increments. As a consequence, this would lead to a weak WWR impact. Second, the time-change approach may also introduce arbitrage opportunities via a forward looking effect. Eventually, even if some techniques exist to deal with WWR in a semi-analytical way (see e.g.~\cite{BF17} and~\cite{FV}), one generally has to rely on Monte Carlo simulations. Standard Monte Carlo methods are known to be computationally intensive. In addition, because of the stochastic clock, the time-change model is very time consuming due to the fact that the simulation is done in a random grid.

Our contribution in this paper is multiple. First, we propose a way to use the time-changed model of Mendoza-Arriaga and Linetsky in the context of CVA avoiding both the correlation destruction and the appearance of arbitrage opportunities. This is achieved by reconstructing the exposure process in a ``synchronous'' way with the intensity, while preserving the original exposure's dynamics. Second, we show via numerical experiments that the corresponding model is indeed able to generate larger WWR CVA figures compared to JCIR++ without facing the inconsistency issue resulting from a negative shift. Eventually, we propose a variance reduction technique based on the {\it adaptive control variate} to reduce the computational cost.

The paper is organized as follows. In Section~\ref{sec2}, we recall the theory of some reduced-form intensity models in the literature of credit risk such as diffusions intensity models and their extensions. In the third section, we introduce the subordinated model combined with a reconstruction of the exposure process avoiding possible arbitrage opportunities resulting from the time-change. The fourth section reviews the basic concepts of CVA computation in the reduced-form setup for the diffusion models and the new subordinated model. In Section~\ref{sect5}, we present the numerical experiments including the comparison of the diffusion models and the time-change model in term of WWR effects and the control variate technique. The last section contains some concluding remarks and perspectives.

\section{Diffusion default intensity Models}
\label{sec2}
Before defining the subordinated model, we recall the definitions of some existing models in the credit risk modelling literature using the intensity approach. In this study, we consider the well known square-root diffusion default intensity models and their extended shifted version SSRD~\cite{BA} and SSRJD~\cite{BE}.
\subsection{The Diffusion intensity Market Model}
We consider a fixed time horizon $T>0$ and a probability space $(\Omega,\mathcal{F}_T,\mathbb{F},\mathbb{Q})$ where $\mathbb{F}=(\mathcal{F}_t)_{0\leq t\leq T}$ is the filtration generated by the vector $\mathbf{W}=(W^B,W^V,W^\perp$). In this setup, $\mathbb{Q}$ represents the risk-neutral probability measure and the components of $\mathbf{W}$ are risk drivers. In particular, $W^B$ governs the dynamics of the risk-free rate $r$, hence that of the bank account num\'eraire:
\begin{equation*}
dB_t = r_tB_tdt,\quad B_0 = 1.
\end{equation*}
Under $\mathbb{Q}$, all prices of tradable assets divided by $B$ are $\mathbb{F}$-martingales between two cash-flow dates. The second Brownian motion $W^V$ drives the dynamics of the portfolio price process 
\begin{equation}
\label{eq5}
dV_t = b(V_t) dt + \sigma (V_t) dW^V_t,\quad V_0>0.
\end{equation}
We assume the coefficients $b,\sigma$ to be regular enough to guarantee that a unique strong solution to this SDE exists. Finally, we model the default time $\tau$ of our counterparty as a random time. It is defined as a first passage time of an increasing stochastic process $\Lambda_t:=\int_0^t\lambda_s ds$, $(\lambda_s)_{s\geq 0}\geq 0$, above a unit-mean exponential random barrier $\mathcal{E}$:

\begin{equation}
\label{eq7}
\tau := \inf\left\{t\geq 0:\,\Lambda_t \geq\mathcal{E}\right\}.
\end{equation}
In this setup, the default intensity $\lambda$ is driven by a  Brownian motion correlated to $W^V$, $W^\lambda:=\rho W^V+\sqrt{1-\rho^2}W^\perp$, $W^\perp \perp W^V$, $\rho\in[-1,1]$, but the threshold $\mathcal{E}$ is independent from $\mathbb{F}$. In such a reduced form setup, the complete filtration $\mathbb{G}=(\mathcal{G}_t)_{0\leq t\leq T}$ is obtained by progressively enlarging $\mathbb{F}$ with $\mathbb{D}=(\mathcal{D}_t)_{0\leq t\leq T}$, the natural filtration of the default indicator $D_t=\indic_{\{\tau \leq t\}}$: $\mathcal{G}_t=\mathcal{F}_t\vee\mathcal{D}_t$ where $\mathcal{D}_t:=\sigma(D_u, 0\leq u\leq t)$. Hence, $\tau$ is a $\mathbb{D}$- and a $\mathbb{G}$-stopping time, but not a $\mathbb{F}$-stopping time. Generally speaking however, $\tau$ and $V$ are related one to another (via $W^V$). And from the Doob-Meyer decomposition of $D$, its $\mathbb{G}$-intensity is given by $\lambda^{\mathbb{G}}_t=(1-D_t)\lambda_t$~\cite{JYC}. Under $\mathbb{F}$, the intensity is simply $\lambda$. Since the random time $\tau$ is constructed through a Cox process, $\mathrm{H}$-\textit{Hypothesis} which state that every $\mathbb{F}$-local martingale is also a $\mathbb{G}$-local martingale holds between the filtrations $\mathbb{F}$ and $\mathbb{G}$, $\mathbb{F}\subset\mathbb{G}$ (from\cite{BJR}, and~\cite{JYC} Proposition $5.9.1.1$ and Remark $7.5.1.2$).

The time-$t$ survival probability to survive up to time $T$ implied by the model is given by 
\begin{equation}
\label{eq8}
P(t,T):=\;\mathbb{Q}\left(\tau > T| \mathcal{G}_t\right)
=\indic_{\{\tau > t\}}\frac{\mathbb{E}[S_T|\mathcal{F}_t]}{S_t}=\indic_{\{\tau > t\}}\frac{G_t(T)}{G_t(t)}
\end{equation}
where $G_t(T):=\mathbb{Q}(\tau > T|\mathcal{F}_t)$ is known as the risk-neutral survival probability in the filtration $\mathbb{F}$ and the \textit{survival process} $S_t:=G_t(t)=\mathbb{Q}(\tau > t|\mathcal{F}_t)$ is the Azéma supermartingale (see~\cite{JLC} for more details). Usually $G_0(T)$ is parametrized as $P^{\rm M}(0,T)=e^{-\int_0^T h(s)ds}$, where $h>0$ is the hazard rate curve prevailing at time $0$. We use the symbol $\mathbb{E}$ to denote expectation under $\mathbb{Q}$.\\
To avoid arbitrage opportunities, one needs to calibrate the model curve to the market curve, i.e. make sure that $P(0,t)=G_0(t)=P^{\rm M}(0,t)$ for all $t>0$.\\
Observe that in the above setup, the survival process takes the simple form $S_t=e^{-\Lambda_t}$. The Azéma supermartingale associated to this type of models has a special Doob-Meyer decomposition: it is decreasing, meaning that the martingale part vanishes. Other types of default models exist for which the martingale part is non-zero, see e.g.~\cite{JV}.

\subsection{The CIR++ and JCIR++ intensity models}
A convenient way to define the intensity process $\lambda$ is to set $\lambda_t=k(X_t)$ where $k$ is a given positive function continuous on $(0,\infty)$ and $X$ follows a Cox-Ingersoll-Ross (CIR) SDE
\begin{equation}
\label{eq6}
dX_t = \kappa(\beta - X_t)dt + \eta \sqrt{X_t}dW^\lambda_t,\quad X_0=x>0.
\end{equation}
By doing so, the intensity process becomes (a function of) a mean-reverting square-root process $X$ with speed of mean reversion $\kappa$, long-term mean $\beta$ and  volatility $\eta$, usually chosen to satisfy the Feller constraint $2\kappa\beta>\eta^2$.

In order to describe the appearance of positive jumps in the default intensity process, we consider the jump-diffusion CIR model (JCIR) defined as
\begin{equation}
\label{eqjcir}
dX_t = \kappa(\beta - X_t)dt + \eta \sqrt{X_t}dW^\lambda_t + dJ_t,\quad X_0=x
\end{equation}
where
\begin{equation}
\label{eqjcomp}
J_t:=\sum_{i=1}^{N_t}Y_i, \quad t\geq 0,
\end{equation}
$N_t$ is a Poison process with intensity $\omega>0$ and $Y_1$, $Y_2,\ldots$ a sequence of identically distributed exponential random variables with mean $1/\alpha$, $\alpha>0$, independent of the Poisson Process $N_t$ and $\mathbf{W}$.

A common choice is to consider $k(x)=x$, in which case the intensity is driven by CIR or JCIR dynamics respectively defined in equations~\eqref{eq6} and \eqref{eqjcir}.
Adding non-negative jumps independent from $W^\lambda$ in the SDE~\eqref{eq6} increases the volatility of the intensity process. These two choices belong to the class of Affine models: the time-$t$ survival probability curve takes the simple form
\begin{equation}
\label{eqab}
P^{{\rm CIR}}(t,T)= \indic_{\{\tau > t\}}A(t,T)e^{-B(t,T)X_t}
\end{equation}
and
\begin{equation}
P^{{\rm JCIR}}(t,T)= \indic_{\{\tau > t\}}\bar{A}(t,T)e^{-\bar{B}(t,T)X_t}
\end{equation}
for some deterministic functions $A,B$, $\bar{A}$ and $\bar{B}$ (see \cite{BM2} for more details). Shifting the process $X$ in a time-dependent way does not affect the above relationship as long as the shift is deterministic but provides full flexibility in terms of calibration capabilities. Therefore, one typically consider $\lambda_t=X_t+\psi(t)$ where $X$ is a CIR or JCIR process and $\psi$ is chosen such that the model and market survival probability curves coincide at inception: $P(0,t)=P^{\rm M}(0,t)$. The corresponding models are know as CIR++ and JCIR++, depending on whether $X$ features jumps or not. The main advantage of adding the shift is that we can fit exactly any term structure of hazard rates and derive analytical formulas both for bonds and European options. In particular, the CIR++ and JCIR++ models remain affine, we just need to replace $A$ by $Ae^{-\int_0^t\psi(s)ds}$ in the CIR model and similarly for $\bar{A}$ in the JCIR model. And the shift $\psi$, at any time $t$, is given by (for both the CIR and JCIR model)
\begin{equation*}
\psi(t) = -\frac{d}{dt}\ln\frac{P^{\rm M}(0,t)}{P(0,t)}.
\end{equation*}
The weakness of this approach is that we can guarantee the positivity of intensities only through restrictions on model parameters such that $\psi\geq 0$. Indeed, $X$ can take values arbitrarily close to zero, so that the condition $\lambda\geq 0$ $\mathbb{Q}$-a.s. is equivalent to saying $\psi\geq 0$. However in general, $\psi$ has to correct the function $P$ so as it sticks, thanks to the shift, to the target function $P^{\rm M}(0,.)$. Given a set of model parameters for $X$, nothing prevents $\psi$ to become negative, in general. But should it take negative values, the resulting model fails to be a Cox-type, and the $\mathrm{H}$-{\it hypothesis} does not hold anymore. This problem is of high importance in practice.
Indeed, the CIR model has low volatility to fit CDS curves, and increasing the volatility just breaks the Feller condition. It is possible to increase the volatility without breaking the Feller constraint using JCIR. However, the affine form of the JCIR model requires the jumps to be independent from the diffusion part, so that positivity allows for upwards jumps only. This of course tends to increase the mean of $\lambda$, and hence to decrease $P^{{\rm JCIR}}$ (the shift can go quickly negative leading to negative intensities which is inconsistent to the Cox model). 
Reciprocally, fitting a given target curve $P^{\rm M}(0,.)$ with non-negative shift only puts constraints on the jump sizes/rates, hence on the attainable volatility.

However, one cannot increase the activity of $J$ without bounds. By doing so indeed, the calibration constraint $P(0,t)=P^{\rm M}(0,t)$ drives the implied shift function $\psi$ downwards. As $\psi$ cannot take negative values, there is a strong limit on the jump rates and/or sizes that one can use while preserving the consistency of the model.\\
In the next section, we propose an alternative model that is less subjected to suffer from the above problems.

\section{Time-change intensity Model}
As explained earlier, the fact that JCIR jumps rapidly pushes $\psi$ downwards results from the fact that the jumps are positive only. This would not be the case if the jumps could go in both directions. Yet, it is not enough just to use symmetric jumps in JCIR++: this would break the positiveness of $X$ if $J$ is independent from $W^\lambda$. 

One possibility consists in modeling $\lambda$ as a time-changed version of a standard intensity process like $X$. On that respect, we define the time-changed CIR (TC-CIR) model by subordinating the CIR process $X_t$ in (\ref{eq6}) with a jump-process
\begin{equation}
\theta_t:=t+J'_t~\footnote{It is possible that $\theta$ is a subordinator with drift $a>0$ (i.e.: $\theta_t=at+J'_t$) but we focus here on the case $a=1$.}
\end{equation}
where $J'$ is a compound Poison process independent of $\mathbf{W}$ defined as in~\eqref{eqjcomp} but with a Poisson process $N'_t$ instead of $N_t$. That is, we define a new process $X^{\theta}$ by $X^{\theta}_t=X_{\theta_t}$ where $\theta$ is the stochastic clock defined above. If the stochastic clock features jumps, the resulting time-changed process would still be positive, and would feature jumps in both directions. This would provide a mean to increase the volatility of $\lambda$ avoiding the implied shift of the time-change model to become negative too quickly as the jumps activity increases. As $X^{\theta}$ is no longer affine, we need to apply the procedure developed by Mendoza-Arriaga  and Linetsky \cite{LIN} to get a closed formula for the survival probability. This approach is a time-changed CIR default intensity by mean of subordination in the sense of Bochner \cite{LIN}. Based on a Cox model, it is analytically tractable by means of eigenfunction expansions of relevant semigroups, yielding closed-form pricing of defaultable zero coupon bonds.

\subsection{The time-changed Market Model}
Consider the corresponding time-changed probability space $(\Omega,\mathcal{F}^\theta_T,\mathbb{F}^\theta,\mathbb{Q})$ with $\mathcal{F}^\theta_t=\mathcal{F}_{\theta_t}$ and $\mathbb{F}^\theta=(\mathcal{F}_{\theta_t})_{t\geq 0}$.
To introduce the time change defaultable market, we consider the default time as defined
in~\eqref{eq7} in order to determine the corresponding intensity of the time-change model. Let's define the corresponding indicator process of $D$ by $D^{\theta}_t:= \indic_{\{\tau\leq \theta_t\}}$, $t\geq 0$. To introduce the time-change filtration, we need first to define an inverse subordinator process $(L_t:=\inf\{s\geq 0:\theta_s>t\}, t\geq 0)$. Let $\mathbb{L}=(\mathcal{L}_t)_{t\geq 0}$ be its completed natural filtration and $\mathbb{H}=(\mathcal{H}_t)_{t\geq 0}$ the enlarged filtration with $\mathcal{H}_t=\mathcal{G}_t\vee\mathcal{L}_t$. We then define our time-changed filtration $\mathbb{H}^{\theta}=(\mathcal{H}^{\theta}_t)_{t\geq 0}$ by $\mathcal{H}^{\theta}_t=\mathcal{H}_{\theta_t}$. Hence, the time-changed bivariate process $(X^{\theta}_t,D^{\theta}_t)_{t\geq 0}$ is $\mathbb{H}^{\theta}$-adapted and c\`adl\`ag and is an $\mathbb{H}^{\theta}$-semimartingale (see \cite{LIN} for details).

In this setup, from the Doob-Meyer decomposition of $D^{\theta}$, our time-changed intensity is (see Theorem $3.3$ (iii) in \cite{LIN}) given by $\lambda^{\mathbb{H}^{\theta}}_t=(1-D^{\theta}_t)\lambda^\theta_t$,
\begin{equation*}
\lambda^\theta_t=k^{\theta}(X^{\theta}_t) \quad \text{with} \quad k^{\theta}(x)=k(x)+\int_{(0,\infty)}\left(1-A(0,s)e^{-B(0,s)x}\right)\nu(ds)
\end{equation*}
where we set $k(x)=x$ (as in the CIR intensity model), $\nu(ds)=\omega\alpha e^{-\alpha s}ds$ and $A,B$ are the same as in~\eqref{eqab}. Hence, if $k^\theta$ is a function from $\mathbb{R}^+$ to $\mathbb{R}^+$ and $X$ is an intensity process, $\lambda^\theta$ defines a new intensity process and can be used to define a new default time using the Cox framework used above:
\begin{equation*}
\tau^\theta := \inf\left\{t\geq 0:\,\int^t_0\lambda^\theta_sds \geq\mathcal{E}\right\}
\end{equation*}
and we have that $\{\tau^\theta\leq t\}\equiv\{\tau\leq\theta_t\}$ $\mathbb{Q}\,$-$\,{\rm a.s.}$, with
\begin{equation*}
D^{\theta}_t= \indic_{\{\tau^\theta\leq t\}}.
\end{equation*}
In this setup, the new time-$t$ survival probability to survive up to time $T$ is
\begin{equation}
P^\theta(t,T):=\;\mathbb{Q}\left(\tau^\theta > T| \mathcal{H}^\theta_t\right)
=\indic_{\{\tau^\theta > t\}}\frac{\mathbb{E}[S^\theta_T|\mathcal{F}^\theta_t]}{S^\theta_t}=\indic_{\{\tau^\theta > t\}}\frac{G^\theta_t(T)}{G^\theta_t(t)}
\end{equation}
where $S^{\theta}_t=\mathbb{Q}(\tau^\theta > t\,|\,\mathcal{F}^{\theta}_t)=e^{-\int^t_0\lambda^\theta_sds}$ is the Azéma supermartingale and $G^{\theta}_t(T) = \mathbb{Q}(\tau^\theta > T\,|\,\mathcal{F}^{\theta}_t)$ the risk-neutral survival probability in the time-change model.

In a multivariate setup in general and in the specific case of CVA application in particular, the time-change approach presents a problem. Indeed, $\lambda^\theta$ is an intensity, which typically can be correlated with other processes (e.g. $V$ and $B$). If we correlate these Brownian motions, two problems arise. First, because of the time change, the correlation between the intensity $\lambda^\theta$ and $(V, B)$ is partially destroyed. Indeed, $V_t$ and $B_t$ depend on $W^V$ and $W^B$ on $[0, t]$ whereas the intensity $\lambda^\theta$ depends on $W^\lambda$ on $[0, \theta_t]$ with $\theta_t\geq t$. This methodology thus impacts negatively the dependence between the processes $\lambda$ and $(B, V)$ as the intensity or the size of jumps of $\theta$ increases.  Another problem, probably even more important, is related to arbitrage opportunities. The knowledge of $\lambda^\theta$ at $t$ contains information on $W^\lambda$ up to $\theta_t$. If $\rho\neq 0$, this introduces a forward looking effect on $V$. It is therefore important to work with a model in which all processes remain synchronized, but without changing the law of $B$ and $V$ as originally specified. To do this, it is enough to rebuild new Brownian motions ($\widetilde{W}^V, \widetilde{W}^B$) so that the increments of $(B, V)$ remain synchronized with those of $\lambda^\theta$.

\begin{lemma}
\label{lem1}
Let $W$ be a Brownian motion and $t_i$ be the time of the $i^{th}$ jump of the Poisson process $N'_t$. Then the process
\begin{equation}
\label{eq12}
\widetilde{W}_t := \sum_{i=0}^{N'_t-1}\int_{t_i\wedge t}^{t_{i+1}^-\wedge t}dW_ {\theta_s} + (W_{\theta_t} - W_{\theta_{t_{N'_t}}})
\end{equation}
is an $\mathbb{F}^{\theta}$-Brownian motion and behaves exactly as $W$ sampled on the time grid.
\end{lemma}
\begin{proof}
$\widetilde{W}_t$ is a continuous $\mathbb{F}^{\theta}$-local martingale with $\widetilde{W}_0=0$ and for every $t\geq 0$, $\langle \widetilde{W}\rangle_t=t$. By the Lévy's characterization theorem, the process $(\widetilde{W}_t)_{t\geq 0}$ is an $\mathbb{F}^{\theta}$-Brownian motion.
\end{proof}
The dynamics of $V$ can thus be equivalently described in terms $\widetilde{W}^V$ setting $W=W^V$ on $\mathbb{F}^{\theta}$ as:
\begin{equation}
\label{eq12bis}
d\widetilde{V}_t= b(\widetilde{V}_t) dt + \sigma (\widetilde{V}_t) d\widetilde{W}^V,\quad \widetilde{V}_0=V_0.
\end{equation}
Applying the same procedure of Lemma \ref{lem1} on $W^B$, the corresponding copy of $W^B$ on the time changed grid is $\widetilde{W}^B$ which will govern the new dynamics of the risk-free rate $\widetilde{r}$ leading to that of the bank account numéraire given by
\begin{equation*}
d\widetilde{B} = \widetilde{r}_t\widetilde{B}dt, \quad \widetilde{B}_0=1.
\end{equation*}
\begin{remark} It is worth stressing the fact that keeping $W^V $ as the driver of the exposure process $V$ (instead of $\widetilde{W}^V$) does not completely destroy the WWR effect resulting from the correlation $\rho$ between $W^V,W^\lambda$. The reason is that even if the \textit{instantaneous correlation between the infinitesimal increments} of $\lambda^\theta$ and $V$ are mutually independent as from the first jump of $\theta$, the correlation between $\lambda^\theta_t$ and $V_t$ is non-zero for any $t>0$ whenever $\rho\neq 0$. This is because these processes depend on the integrals of increments of Brownian motions on some time intervals. Between two jumps of the clock in particular, the Brownian increments driving the change in the intensity process $\lambda^\theta$ can be independent from the Brownian increments driving the exposure $V$ on a same time period, but the increments of the first Brownian motion on a given time interval can be dependent on the second Brownian motion on \textit{another} interval. This explains that two processes $\lambda^\theta$ and $V$ can be dependent on each other even if the instantaneous correlation of their increments vanishes because of the time-change. By contrast, the correlation between $\lambda^\theta_t$ and $V_t$  is lower than that of $\lambda^\theta_t$ and $\widetilde{V}_t$: by synchronizing the changes of the Brownian motion driving $V$ with the one driving $\lambda^\theta$, one maximizes the attainable correlation. For instance, let $W,B$ be two Brownian motions with instantaneous correlation $\rho$, i.e. $d\langle W,B\rangle_t=\rho dt$ and define $B^\delta$ as $B^\delta_t=B_{t+\delta}$ where $\delta>0$. The correlation between the increments of $W$ and $B$ on the interval $[s,t]$ is $\rho$ whereas that between $W$ and $B^\delta$ is $\rho (t-(s+\delta))^+/(t-s)$ whose absolute value is no greater than $|\rho|$.
\end{remark}
\begin{theorem}
\label{th1}
Discounted payoffs driven by $\widetilde{V}$ and $\widetilde{B}$ are $\mathbb{H}^\theta$-martingales under $\mathbb{Q}$.
\end{theorem}
\begin{proof}
We know that the discounted payoffs driven by $V$ and $B$ are $(\mathbb{Q}, \mathbb{F})$-martingales, hence they are $(\mathbb{Q}, \mathbb{G})$-martingales due to the $\mathrm{H}$-\textit{Hypothesis}. By construction, $\widetilde{V}$ and $\widetilde{B}$ have the same $\mathbb{F}^\theta$ dynamics as $V$ and $B$ under $\mathbb{F}$ (see Lemma \ref{lem1}). Hence,  the discounted payoffs driven by $\widetilde{V}$ and $\widetilde{B}$ are $(\mathbb{Q},\mathbb{F}^\theta)$-martingales. Because $\tau^\theta$ is modelled with a Cox process, immersion holds and they remains martingales when progressively enlarging $\mathbb{F}^\theta$ with $\tau^\theta$. This shows that they are $(\mathbb{Q},\mathbb{G}^\theta)$-martingales, hence $(\mathbb{Q},\mathbb{H}^\theta)$-martingales (since $\mathcal{L}_{\theta_t}=t$).
\end{proof}

As explained earlier, time-changing the intensity can introduce a forward-looking effect and thus arbitrage opportunities. Indeed, prices of non-dividend paying asset discounted at the risk-free rate would not be a martingale under $\mathbb{Q}$ in the natural filtration generated by $(W^V,W^{\lambda^\theta})$ when $\rho\neq 0$. This is formalized in the next lemma proven in the Appendix.
\begin{lemma}
\label{cor1}
Suppose that $V$ is a martingale with respect to $\mathbb{F}^{W^V}$, the natural filtration of $W^V$. It is also a martingale with respect to $\mathcal{F}^{W^V}\vee\mathcal{F}^{W^\lambda}$. We note $W^{\lambda^\theta}$ the time-changed version of $W^\lambda$. For instance suppose that for some $s\in[0,t]$, we have $s<\theta_s<t$. Then, $X$ may not be a $\mathcal{F}^{W^V}\vee\mathcal{F}^{W^{\lambda^\theta}}$-martingale.
\end{lemma}

\subsection{The time-changed CIR++ intensity model}
To define the shifted time-change CIR (TC-CIR++) model, we need to have a closed form of the survival probability of the TC-CIR model. For that, we only need to compute the Laplace transform of our time-change process $\theta$ in order to apply the idea devised in \cite{LIN}.

\subparagraph{Laplace transform of a Lévy subordinator:}

Our time-change jump-process $\theta$ is a Lévy subordinator and its Laplace transform can be found easily using the Lévy-Khintchine formula~\cite{DA}. For any $u\in \mathbb{R}$, the Laplace transform of $\theta$ reads
\begin{equation*}
\mathbb{E}[e^{-u\theta_t}]=\mathbb{E}[e^{-u(t+J_t)}]= e^{-u t}\mathbb{E}[e^{-u J_t}]\;.
\end{equation*}
From the exponential formula, a Lévy-Khintchine formula representation of the compound Poisson process yields
\begin{equation*}
\mathbb{E}[e^{-u J_t}]=e^{-t\varphi(u)}
\end{equation*}
where $\varphi$ is the Lévy exponent given by
\begin{equation*}
\varphi(u)=\int_{\mathbb{R}` \{0\}}(1-e^{-u s})\omega\alpha e^{-\alpha s}ds\indic_{\{s>0\}}=\int_{(0,\infty)}\omega\alpha(1-e^{-u s}) e^{-\alpha s}ds.
\end{equation*}
It comes that
\begin{equation*}
\mathbb{E}[e^{-u\theta_t}]=e^{-t\phi(u)}~~,~~\phi(u)=u\left(\frac{u+\alpha+\omega}{u+\alpha}\right).
\end{equation*}
Knowing $\phi$, the time-changed survival probability takes the closed form
\begin{equation*}
P^{\theta}(t,T)=\indic_{\{\tau^\theta>t\}}\sum^{\infty}_{n=1}e^{-\phi(\lambda_n)(T-t)}f_n(0)\varphi_n(X^{\theta}_t)\\
\end{equation*}
where $\lambda_n$, $f_n$ and $\varphi_n$ are given in \cite{LIN}.

The TC-CIR++ model is obtained by defining the time-changed intensity process as $\lambda^{\theta}_t = k^{\theta}(X^{\theta}_t)+\psi^\theta(t)$ and finding $\psi^\theta$ such that $P^{\theta}(0,T)=G^{\theta}_0(T)=P^{\rm M}(0,T)$. Hence
\begin{equation*}
\psi^\theta(t) = -\frac{d}{dt}\ln\frac{P^{\rm M}(0,t)}{P^\theta(0,t)}.
\end{equation*}

\section{CVA in a reduced form setup}
\label{sect4}
The reduced form approach relies on a change of filtrations. In this section, we derive the CVA formulas in both cases where the default intensity is given by the square-root diffusions or the time-change model.

CVA attempts to measure the expected loss due to missing the remaining payments of the OTC portfolio. Its mathematical expression is given in a risk-neutral pricing framework based on a no-arbitrage setup. Let $R$ be the recovery rate of the counterparty and the exposure processes $V$ and $\widetilde{V}$ respectively given by~\eqref{eq5} and~\eqref{eq12bis}. For the mathematical proof of the following CVA expressions, we refer to~\cite{JYC}.
\subsection{CVA formula in the diffusion model}
From the $\mathrm{H}$-\textit{Hypothesis}, payoffs driven by $V$ and $B$ are $\mathbb{G}$-martingale. In addition, since $V/B$ is $\mathbb{Q}$-integrable and $\mathbb{F}$-predictable, assuming $\tau>0$ and deterministic recovery rate $R$, the time-$t$ CVA expression reads

\begin{equation}
\label{eq2}
\mathrm{CVA}_t = \indic_{\{\tau>t\}}B_t\mathbb{E}\left[(1-R)\frac{V^+_{\tau}}{B_\tau}\indic_{\{\tau\leq T\}} \bigg | \mathcal{G}_t\right]=-\indic_{\{\tau>t\}}\frac{B_t}{S_t}\mathbb{E}\left[(1-R)\int_t^T\frac{V^+_u}{B_u} dS_u\bigg |\mathcal{F}_t\right]\;.
\end{equation}
Discretizing the integral with a numerical scheme, the Monte-Carlo estimation of time-$0$ CVA becomes
\begin{equation}
\label{eq2bis}
\widehat{\mathrm{CVA}}_0 := -(1-R)\frac{1}{m}\sum_{i=1}^m\sum_{k=1}^n \frac{V^{+,(i)}_{t_k}}{B^{(i)}_{t_k}}\Delta S^{(i)}_{t_k},\quad n=\frac{T}{\delta}
\end{equation}
where $\Delta S^{(i)}_{t_k}=S^{(i)}_{t_k}-S^{(i)}_{t_{k-1}}$. The right-hand side results from Monte Carlo approximation, by taking the sample mean of $m$ time-integrals discretized in $n$ intervals of length $\delta$.

In the specific case where $\tau$ is independent from the discounted exposure (i.e. $\rho=0$), the {\it independent} CVA formula is given by
\begin{equation}
\label{eq4}
\mathrm{CVA}_0^{\perp} = -(1-R)\int_0^T\mathbb{E}\left[\frac{V^+_u}{B_u}\right]dG_0(u).
\end{equation}
In other words, CVA only depends \textit{separately} on the expected discounted exposure $\mathbb{E}\left[\frac{V^+_u}{B_u}\right]$ and the prevailing mrisk-neutral survival probability curve $G_0(.)$. We refer to~\cite{BF17} for more details.

Generally speaking however the later expression does not hold, and CVA depends on the \textit{joint dynamics} of the exposure and credit worthiness. We cannot get to such a simpler formula as in (\ref{eq4}) and we have to take in account the dependency between credit and exposure. Wrong way risk (WWR) is the additional risk related to this dependency.

\subsection{CVA formula in the time-change model}
As $\widetilde{V}/\widetilde{B}$ is $\mathbb{Q}$-integrable and $\mathbb{F}^{\theta}$-predictable, assuming $\tau^\theta>0$ and using Theorem \ref{th1}, in a similar ways as before, but using $(\widetilde{V},\widetilde{B})$ instead of $(V,B)$ (having the same dynamics under $\mathbb{F}^\theta$ as $(V,B)$ under $\mathbb{F}$), but using the intensity of $\tau^\theta$ with Azéma supermartingale $S^\theta$, we have
\begin{equation}
\label{eqCVATC}
\begin{aligned}
\mathrm{CVA}_t&=\indic_{\{\tau^\theta>t\}}\widetilde{B}_t\mathbb{E}\left[(1-R)\frac{\widetilde{V}^+_{\tau}}{\widetilde{B}_{\tau}}\indic_{\{\tau^\theta\leq T\}} \bigg |  \mathcal{H}^{\theta}_t \right]\\
&=\indic_{\{\tau^\theta>t\}}\widetilde{B}_t\mathbb{E}\left[(1-R)\frac{\widetilde{V}^+_{\tau}}{\widetilde{B}_{\tau}}\indic_{\{\tau^\theta\leq T\}} \bigg |  (\tau^{\theta}\leq t)\vee\mathcal{F}^{\theta}_t \right]\\
&=-\indic_{\{\tau^\theta>t\}}\,\frac{\widetilde{B}_t}{S^{\theta}_t}\,\mathbb{E}\left[(1-R)\int_t^T \frac{\widetilde{V}^+_u}{\widetilde{B}_u}dS^{\theta}_u \bigg |  \mathcal{F}^{\theta}_t \right]\;.
\end{aligned}
\end{equation}
Therefore, the time-$0$ CVA in the time-changed model can be approximated using $m$ paths of Monte Carlo simulations as
\begin{equation}
\label{eq13}
\widehat{\mathrm{CVA}}_0 := -(1-R)\frac{1}{m}\sum_{i=1}^m\sum_{k=1}^n\frac{\widetilde{V}^{+,(i)}_{t_k}}{\widetilde{B}_{t_k}}\Delta S^{\theta,(i)}_{t_k},\quad n=\frac{T}{\delta}\;.
\end{equation}
In the specific case where $\tau^\theta$ is independent from the discounted exposure, we obtain the {\it independent} CVA formula in the time-changed model
\begin{equation}
\label{eq4bis}
\mathrm{CVA}_0^{\perp} = -(1-R)\int_0^T\mathbb{E}\left[\frac{\widetilde{V}^+_u}{\widetilde{B}_u}\right]dG^{\theta}_0(u)\;.
\end{equation}
Observe that by construction of $\psi$ and $\psi^\theta$, $G_0(t)=G^\theta_0(t) = P^{\rm M}(0,t)$. Hence, ${\rm CVA}_0^\perp$ agrees in either models under the calibration constraint. This means
\begin{equation}
\label{eqInd}
\mathrm{CVA}^{\perp}_0 = -(1-R)\int_0^T\mathbb{E}\left[\frac{V^+_u}{B}_u\right]dP^{\rm M}(0,u)=-(1-R)\int_0^T\mathbb{E}\left[\frac{\widetilde{V}^+_u}{\widetilde{B}_u}\right]dP^{\rm M}(0,u).
\end{equation}

\section{Numerical experiments}
\label{sect5}
In this section, we start by defining the simulation procedure of the bivariate process $(X^\theta,\widetilde{V})$. The CVA is computed using standard Monte Carlo simulation and the performance of the shifted time-changed model in term of WWR is compared to the CIR++ and JCIR++ stochastic intensity models. For the sake of simplicity, the recovery rate $R$ and the interest rate $r$ are assumed to be constant and set to zero (i.e. $R=0$ and $B=\widetilde{B}=1$) to put the focus and the treatment of the credit-exposure dependency.
\subsection{Simulation procedure of ($X^\theta, \widetilde{V}$)}
 Let's denote by $\mathcal{T}=\left\lbrace 0,\delta,2\delta,\ldots,T\right\rbrace$ the time-$t$ grid and $\mathcal{T}^{\theta}=\{0,\theta_\delta,\theta_{2\delta},\ldots,\theta_T\}$ the grid at time $\theta_t$. To refine the grid $\mathcal{T}^{\theta}$, we define grids $\mathcal{T}^{\theta}_i,i=1,\ldots,n$, as\\
\begin{equation*}
\mathcal{T}^{\theta}_i:=\overset{n_i-1}{\underset{j=1}{\bigcup}}\left\lbrace\theta_{t_{i-1}} + j\frac{\Delta\theta_{t_i}}{n_i} \right\rbrace,\quad n_i=\left\lceil\frac{\Delta\theta_{t_i}}{\delta}\right\rceil\quad,\quad \mathcal{T}^{\theta}_{fine}:=\mathcal{T}^{\theta}\bigcup\left\lbrace\ \overset{n}{\underset{i=1}{\bigcup}}\mathcal{T}^{\theta}_i \right\rbrace
\end{equation*}
The simulation grid $\mathcal{T}^{\theta}_{fine}$ contains $\mathcal{T}^{\theta}$ and is completed in such a way that (after sorting), the step between two consecutive points is no greater than the chosen time step $\delta$ to keep control on the discretization error independently of the jump sizes of $\theta$.\\
To simulate ${X}^{\theta}$, we simulate the CIR process $X$ in (\ref{eq6}) on $\mathcal{T}^{\theta}_{fine}$ using Diop's scheme~\cite{DI}:
\begin{equation*}
\bar{X}_{(i+1)\delta} = \bar{X}_{i\delta} + \kappa(\beta - \bar{X}_{i\delta}^+)\delta + \eta\sqrt{\bar{X}_{i\delta}^+}(W^\lambda_{(i+1)\delta} - W^\lambda_{i\delta}),\quad \bar{X}_0=X_0,\quad i=0,\ldots,n-1.
\end{equation*}
${X}^{\theta}$ is obtained by extracting in $\mathcal{T}^{\theta}_{fine}$ the corresponding values of $X$ on the grid $\mathcal{T}^{\theta}$.\\
To obtain $\widetilde{V}$, we simulate $W^V$ on $\mathcal{T}^{\theta}_{fine}$, extract its corresponding values on $\mathcal{T}^{\theta}$ and use respectively (\ref{eq12}) and (\ref{eq12bis}).

\subsection{Numerical results}
\label{subsec1}
In this section, we compare the performances of the three models studied above in terms of WWR impact for a simple forward-type Gaussian exposure and Brownian swap bridge exposure. We fix the CIR parameters $(\kappa,\beta,\eta,x)$ as in~\cite{BF17} and search for the jump parameters $(\omega,\alpha)$ of JCIR (affecting directly the intensity) and  of TC-CIR (affecting the stochastic clocks and only indirectly the jumps in the intensity) such that $\psi\geq 0$ and the WWR impact is maximum.
\subsubsection{Brownian exposure}
We set the coefficients of the exposure dynamics in eq. \eqref{eq5} to $b(V_t)\equiv0$, $\sigma(V_t)\equiv\sigma$ and $V_0=0$ which leads to
\begin{equation*}
dV_t=\sigma dW^V_t.
\end{equation*}
This example illustrates a $3$ years forward contract or total return swap, which does not pay dividends. In Figure~\ref{fig:FigA} we plot the CVA in function of the correlation $\rho$ for the tree considered models. We notice that due the the shift constraint, the JCIR model has slightly more WWR impact than the CIR model (because of the jumps) while the TC-CIR model has a larger WWR impact.
\begin{figure}[H]
\centering
\subfigure[CIR$\,(0.02,0.161,0.08,0.03),\,$JCIR$\,(0.07,0.08)$, TC-CIR$\,(0.6,0.512)$]{\includegraphics[width=0.48\columnwidth]{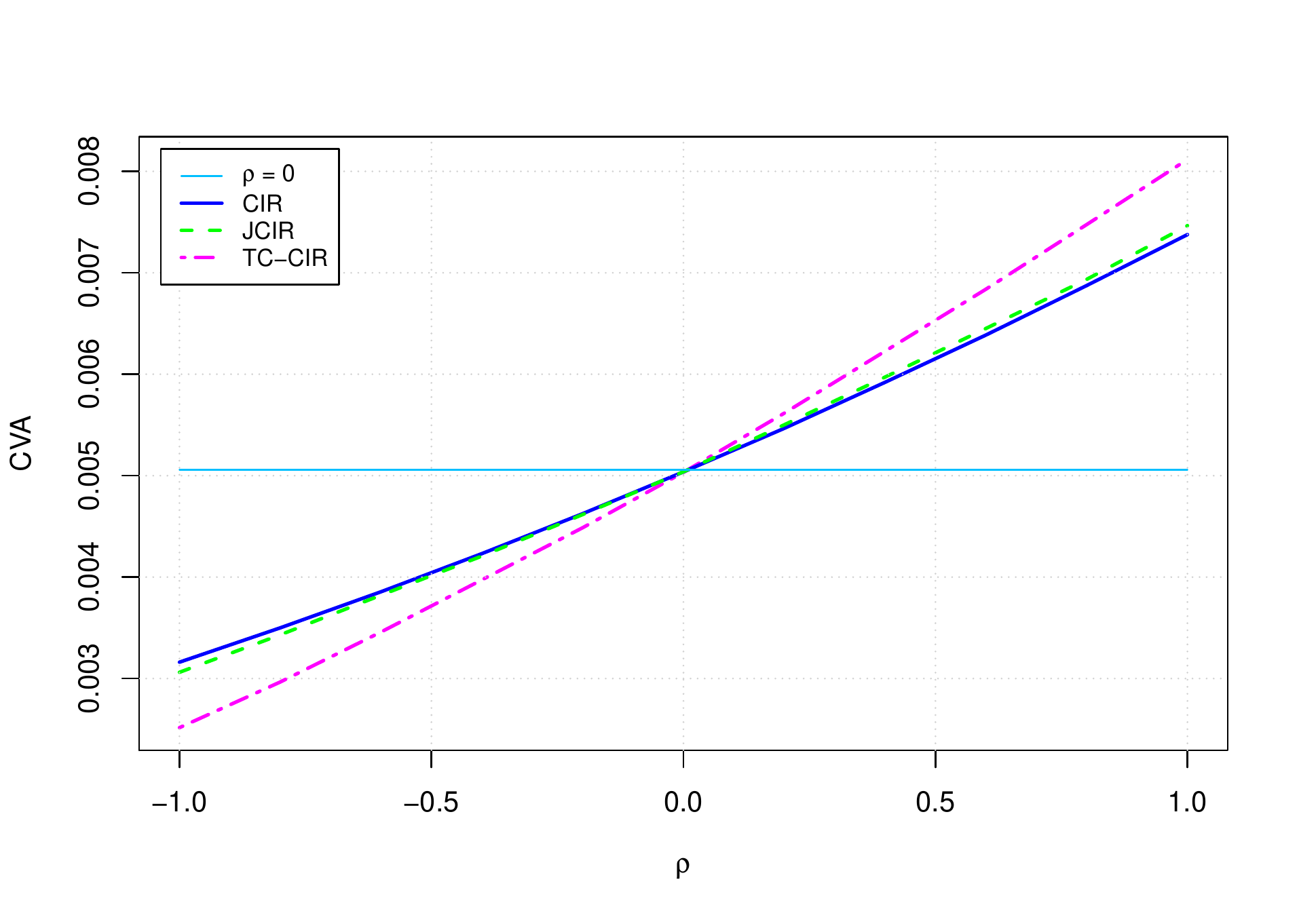}}\hspace{0.25cm}
\subfigure[CIR$\,(0.072,0.05,0.045,0.04),\,$JCIR$\,(0.07,0.05)$, TC-CIR$\,(0.4,0.49)$]{\includegraphics[width=0.48\columnwidth]{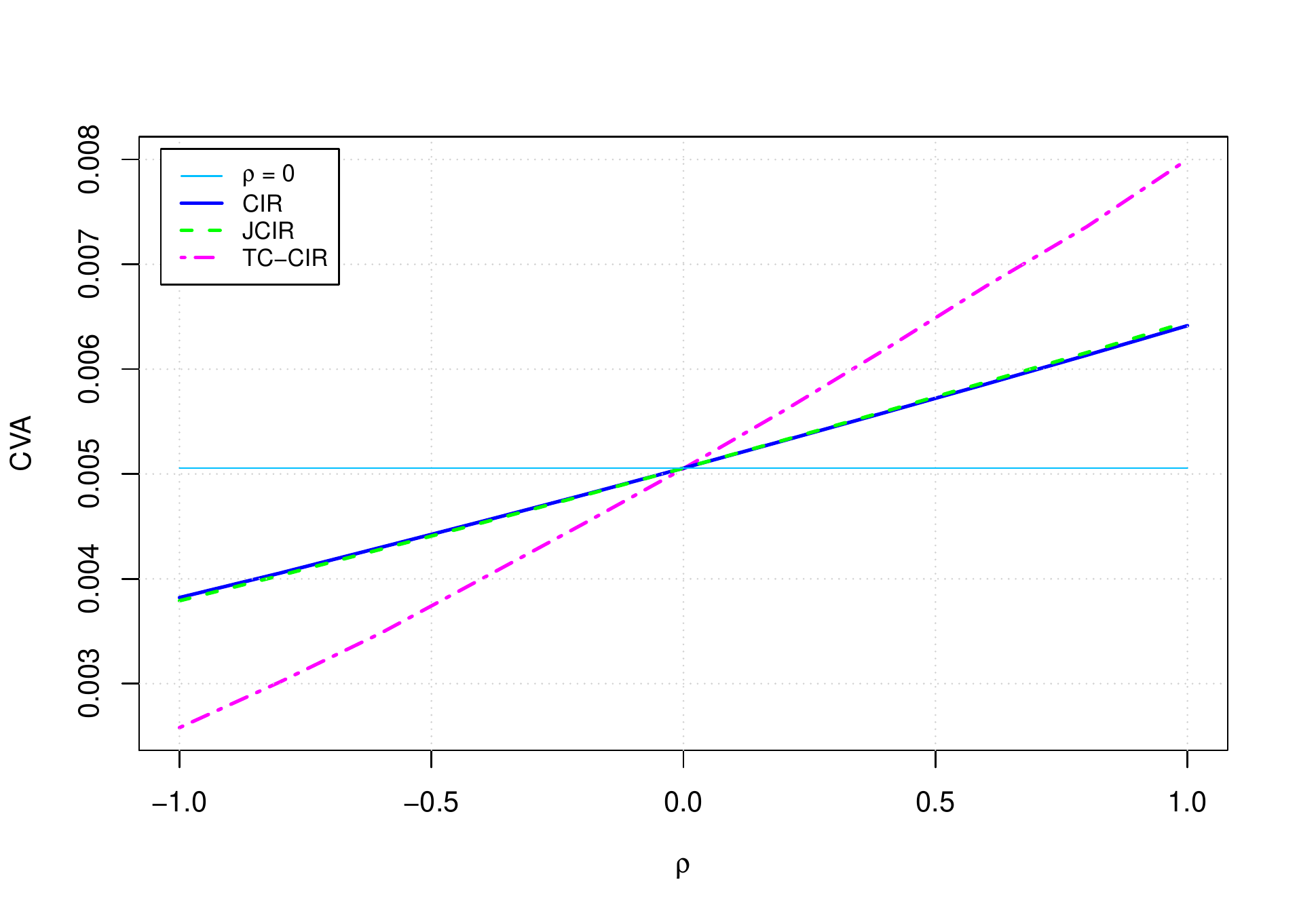}}
\caption{CVA figures, 3Y Gaussian exposure, $\sigma=8\%$. Hazard rate is $h(t)=5\%$.}\label{fig:FigA}
\end{figure}

\subsubsection{Brownian bridge (with drift) exposure}
The drifted Brownian bridge is obtained by choosing in \eqref{eq5} $b(V_t)\equiv\gamma(T-t)-\frac{V_t}{T-t}$, $\sigma(V_t)\equiv\sigma$ and $V_0=0$, so that
\begin{equation*}
dV_t=\left[\gamma(T-t)-\frac{V_t}{T-t}\right]dt + \sigma dW^V_t
\end{equation*}
where $\gamma$ stands for the future expected moneyness of an interest rate swap implied by the forward curve and $\sigma$ controls the exposure volatility.
We get similar results in term of WWR impact as in the previous example.
\begin{figure}[H]
\centering
\subfigure[CIR$\,(0.02,0.161,0.08,0.03),\,$JCIR$\,(0.07,0.08)$, TC-CIR $(0.6,0.512)$ ]{\includegraphics[width=0.48\columnwidth]{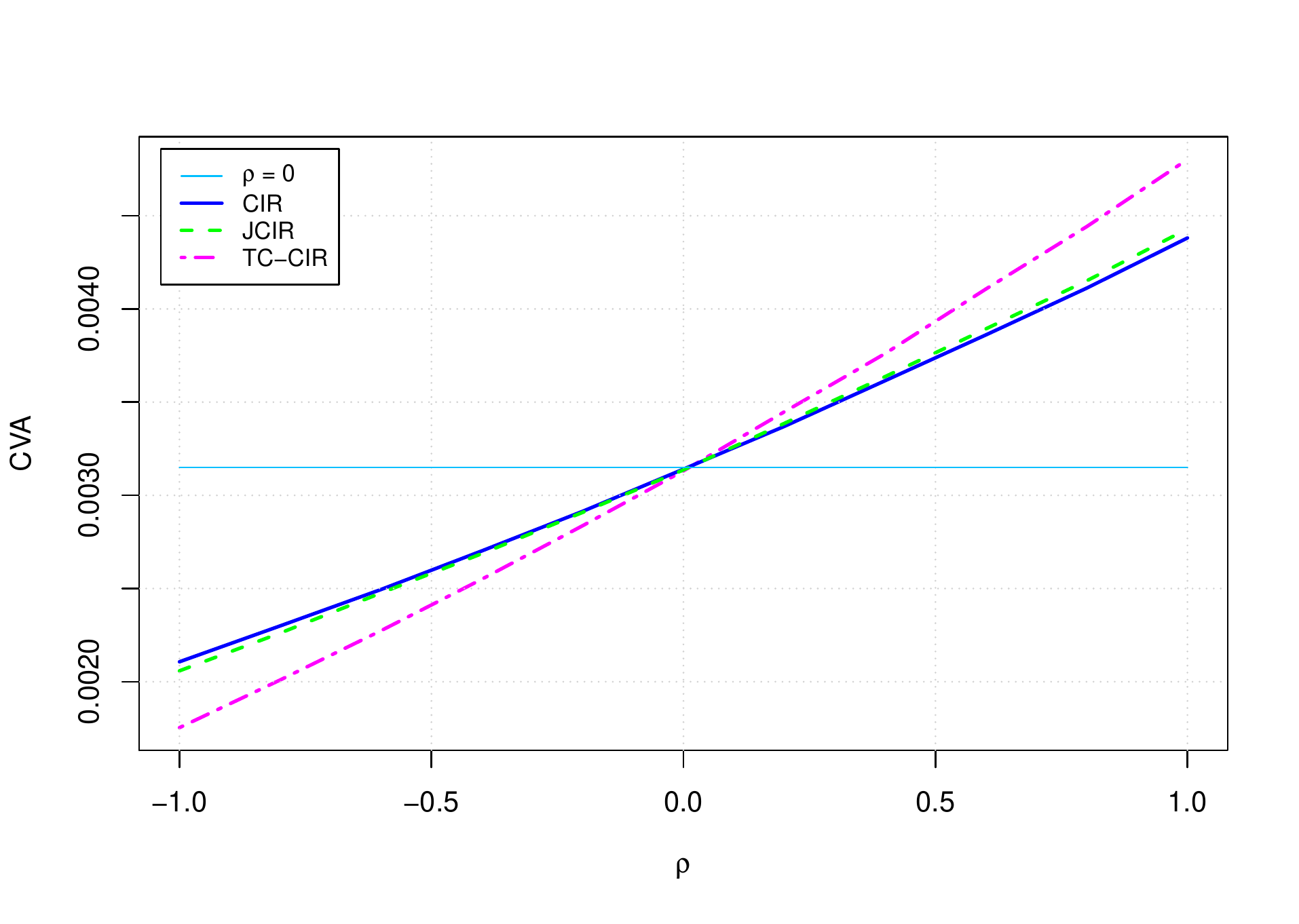}}\hspace{0.25cm}
\subfigure[CIR$\,(0.072,0.05,0.045,0.04),\,$JCIR$\,(0.07,0.05)$, TC-CIR  $(0.4,0.49)$]{\includegraphics[width=0.48\columnwidth]{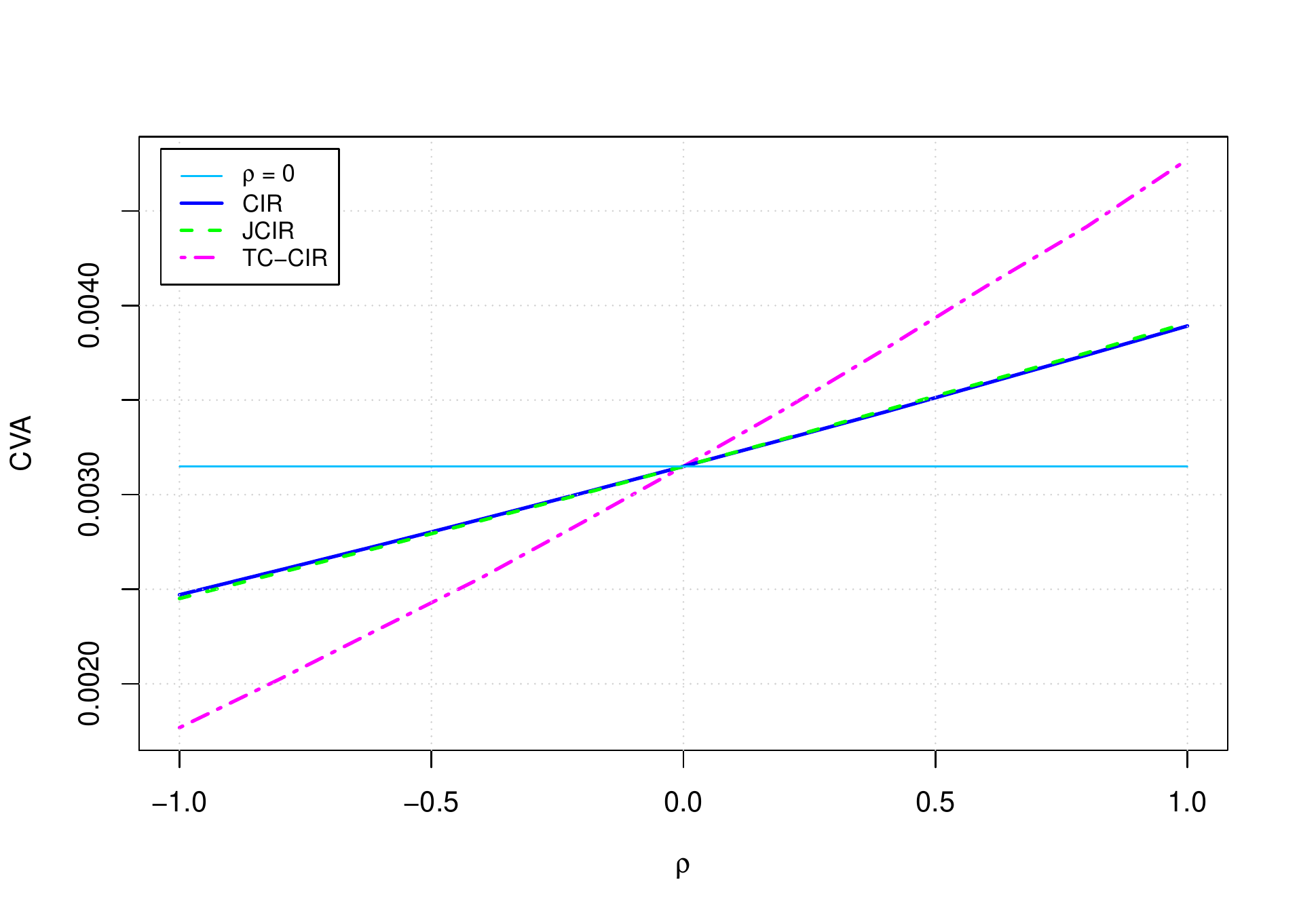}}
\caption{CVA figures, 3Y Swap exposure, $(\sigma, \gamma)=(8\%,0.1\%)$. Hazard rate is $h(t)=5\%$.}\label{fig:FigB}
\end{figure}

\subsection{Adaptive control variate Monte Carlo estimator}

The standard Monte Carlo method applied to the TC-CIR++ model is time consuming. The reason is that the time-step needs to be kept relatively small (to limit the discretization errors) whereas the simulation horizon is governed by $\theta_T$ which can be much larger than $T$. In order to reduce the computational cost, we propose to adopt a variance reduction technique called \textit{adaptive control variate}. In this section, we briefly recall the idea of control variate, describe its adaptive implementation and then transpose it to our CVA application.

Our purpose is to find an estimator of $\mathbb{E}[Y]$ from $m$ i.i.d. observations of $Y$. The unbiased sample-mean estimator of $\mathbb{E}[Y]$ is $\hat{Y}:=\frac{1}{m}\sum_{k=1}^m Y_k$. The idea of control variate consists of finding an alternative unbiased estimator $\tilde{Y}$ with lower variance compared to $\hat{Y}$ by using a control variate $Z$ with known expectation. Consider a generic pair $(Y,Z)$ of random variables with i.i.d copies $(Y_k,Z_k),$ $k\in\{1,2,\ldots,m\}$ and define
%
\begin{equation}
\label{eqBatch}
Y_k^\mu:=Y_k-\mu \Xi_k~~~,~~ \Xi_k:= Z_k- \mathbb{E}[Z]\;.
\end{equation}
The sample-mean estimator of $Y_k^\mu$ is an alternative unbiased estimator of $\mathbb{E}[Y]$, and is given by
\begin{equation*}
\hat{Y}^\mu = \frac{1}{m}\sum_{k=1}^m Y_k^\mu =\frac{1}{m}\sum_{k=1}^m(Y_k-\mu\Xi_k)\;.
\end{equation*}
Its variance is equal to that of $\hat{Y}$ for $\mu=0$ but is minimum for $\mu=\mu^*$ where
\begin{equation*}
\mu^* := \frac{{\rm Cov}(Y, \Xi)}{{\rm Var}(\Xi)}=\frac{\mathbb{E}[ Y \Xi]}{\mathbb{E}[\Xi^2]}\;. 
\end{equation*}
Because ${\rm Var}(\hat{Y}^{\mu^*})=(1-{\rm Corr}^2(Y, \Xi)){\rm Var}(\hat{Y})$, this approach is interesting when choosing the control variate $\Xi$ highly correlated with $Y$.

In practice however, the optimal constant $\mu^*$ needs to be itself estimated. The \textit{adaptive} control variate uses a different value for $\mu$ for every index $k\in\{1,2,\ldots,m\}$ :
\begin{equation*}
V_k:=\frac{1}{k}\sum_{i=1}^k\Xi_i^2,\quad C_k:=\frac{1}{k}\sum_{i=1}^kY_i\Xi_i \quad {\rm and} \quad \mu_k:=\frac{C_k}{V_k}\;,
\end{equation*}
where $\mu_0:=0$ and $\mu_{k-1}$ is the best estimator of $\mu^*$ at step $k$ (see~\cite{GP} for more details). Eventually, the {\it adaptive control variate} estimator of $\mathbb{E}[Y]$ is given by
\begin{equation*}
\tilde{Y}^\mu:=\frac{1}{m}\sum_{k=1}^mY_k^{\mu_{k-1}}=\hat{Y}-\frac{1}{m}\sum_{k=1}^m \mu_{k-1} Z_k +\frac{\mathbb{E}[Z]}{m}\sum_{k=1}^m \mu_{k-1}\;.
\end{equation*}

In our CVA application, we are interested in estimating $\mathrm{CVA}_0$ which is nothing but $\mathbb{E}[Y]$ with $Y = -\int_0^TV^+_u dS_u$ in the CIR and JCIR models (\ref{eq2}) or $Y = -\int_0^T\widetilde{V}^+_u dS^\theta_u$ in the TC-CIR model \eqref{eqCVATC}. We take as control variable $Z = -\int_0^TV^+_u dS^\perp_u$ or $Z = -\int_0^T\widetilde{V}^+_u dS^{\theta,\perp}_u$, respectively, where $S^\perp$ (resp. $S^{\theta,\perp}$) is the survival process $S$ (resp. $S^\theta$) associated to the intensity $\lambda$ (resp. $\lambda^\theta$) simulated using $W^\perp$ instead of $W^\lambda$.\footnote{Alternatively, if the trajectories of $V$ and $S$ are stored, it is enough to shuffle those of $V$ (or $S$), so as to combine, in the CVA computation, the $i$-th exposure's sample path with the $\pi(i)\neq i$ survival process' sample path.} In light of the above development, this choice is appealing because $Z$ is correlated with $Y$ (via the exposure process as well as the $W^\perp$ component of the survival process) whereas the expectation of $Z$ is known in closed form and corresponds to $\mathrm{CVA}^{\perp}_0$ given in \eqref{eqInd}. In the adaptive procedure, the $m$ pairs $(Y_k,Z_k)$ are given by integrals $(Y,Z)$ over the corresponding scenario. They are i.i.d. copies of $Y$ and $Z$ (up to discretization error resulting from the integral computation and simulation scheme). Eventually, our \textit{adaptive} CVA estimator for the CIR and JCIR models reads
\begin{equation}
\widetilde{\rm CVA}_0 = \widehat{\rm CVA}_0 - \frac{{\rm CVA}_0^\perp}{m}\sum_{k=1}^m \mu_{k-1} +\frac{1}{m}\sum_{k=1}^m\sum_{j=1}^n \mu_{k-1} V^{+,(k)}_{t_j}\Delta S^{(k)}_{t_j}
\end{equation}
with $\widehat{\rm CVA}_0$ given in \eqref{eq2bis} and \eqref{eq13}, respectively. The TC-CIR estimator takes a similar form provided that one replaces $(V,S)$ by $(\widetilde{V},S^\theta)$.

In figure~\ref{fig:FigC}, we compare the confidence interval at level $95\%$ of the CVA computation resulting from standard Monte Carlo (MC) and adaptive control variate (CV). We apply this technique for the three considered models (CIR, JCIR and TC-CIR) by using the same set of parameters in figure~\ref{fig:FigA} (a) with a Gaussian exposure. Similar results are detailed in the Appendix when using the same Gaussian exposure profile but with the parameters of figure~\ref{fig:FigA} (b). 
\begin{figure}[H]
\centering
\subfigure[CIR$\,(0.02,0.161,0.08,0.03)$]{\includegraphics[width=0.48\columnwidth]{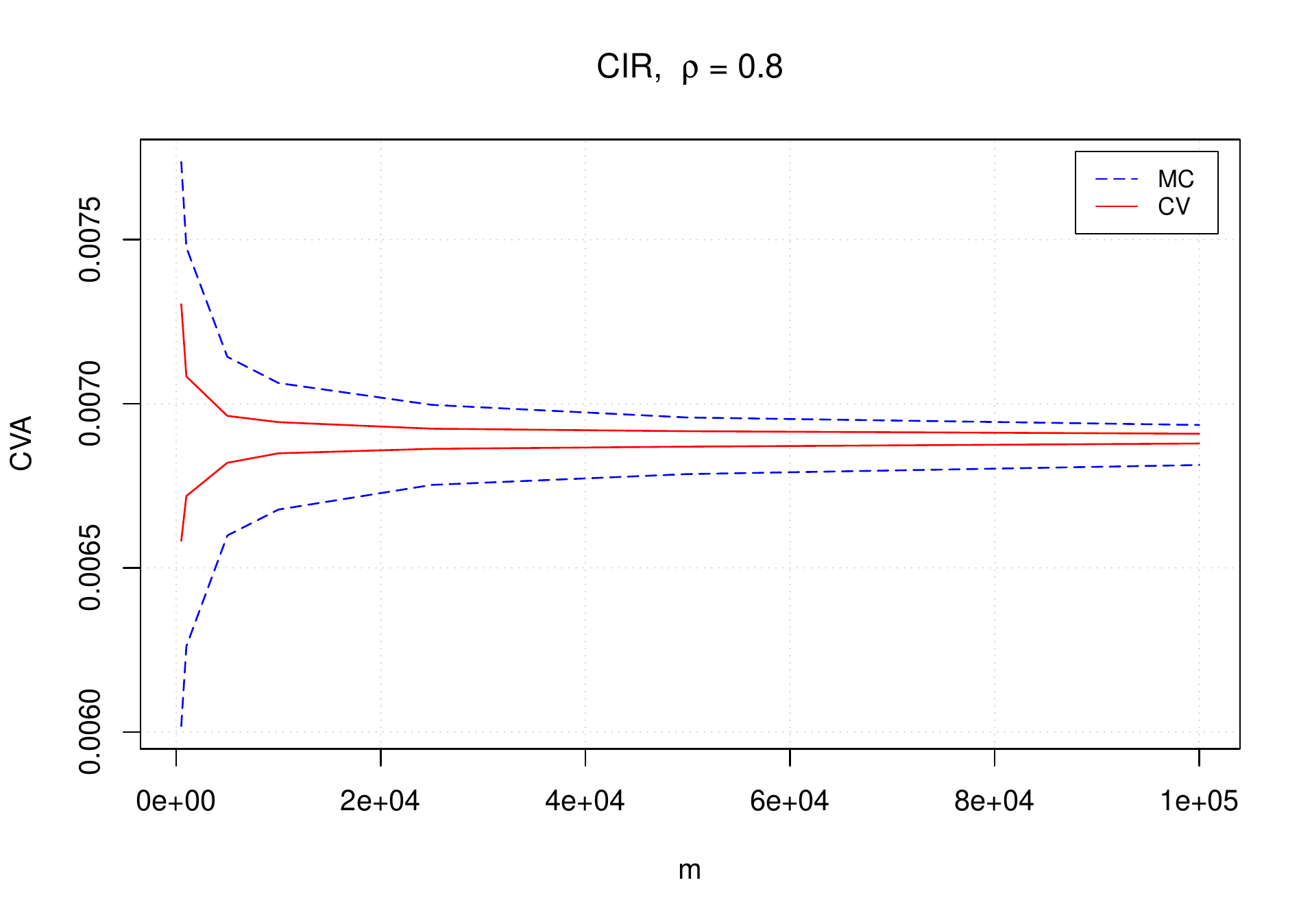}}\hspace{0.25cm}
\subfigure[CIR$\,(0.02,0.161,0.08,0.03)$]{\includegraphics[width=0.48\columnwidth]{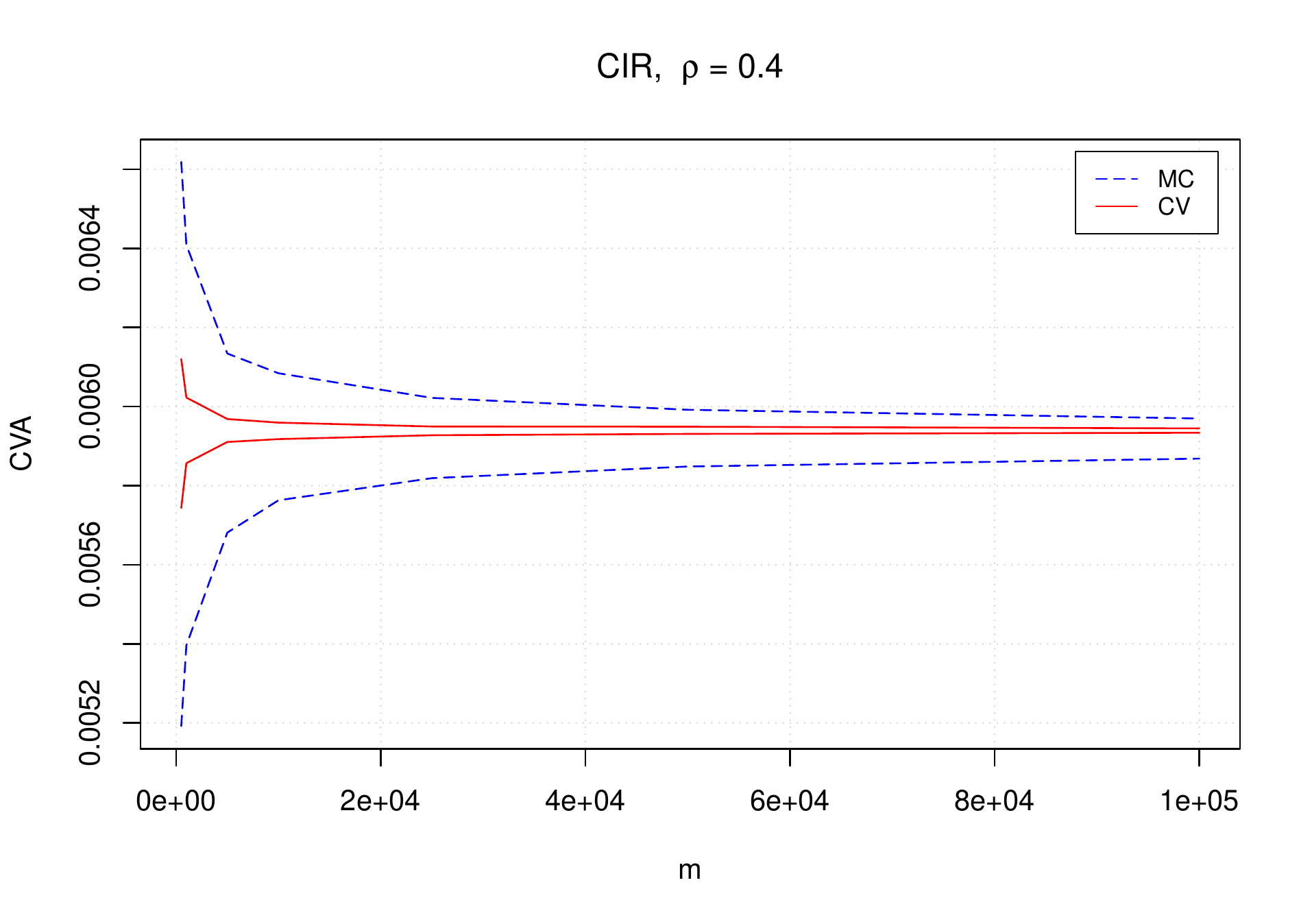}}
\subfigure[CIR$\,(0.02,0.161,0.08,0.03),\,$JCIR$\,(0.07,0.08)$]{\includegraphics[width=0.48\columnwidth]{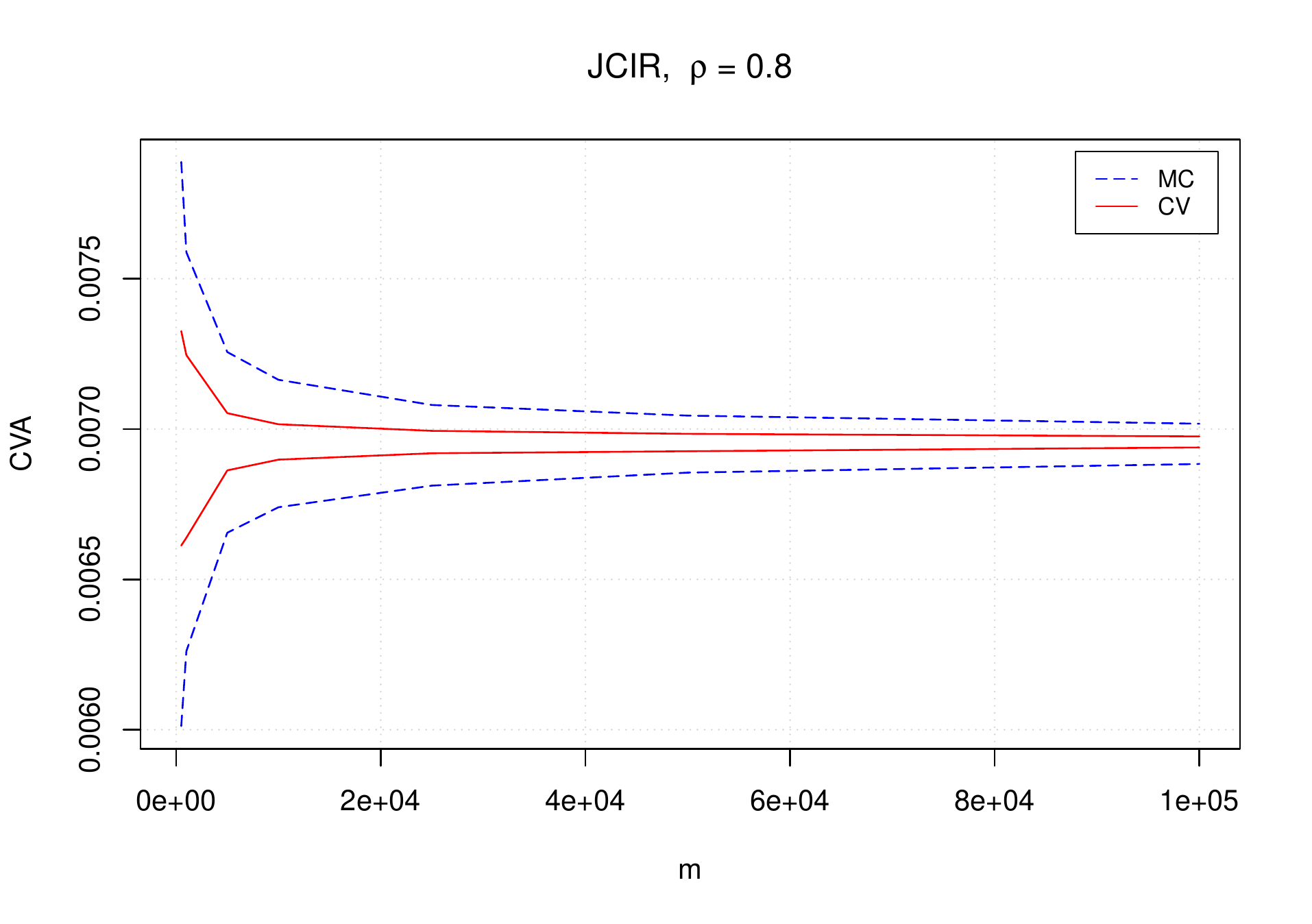}}\hspace{0.25cm}
\subfigure[CIR$\,(0.02,0.161,0.08,0.03),\,$JCIR$\,(0.07,0.08)$]{\includegraphics[width=0.48\columnwidth]{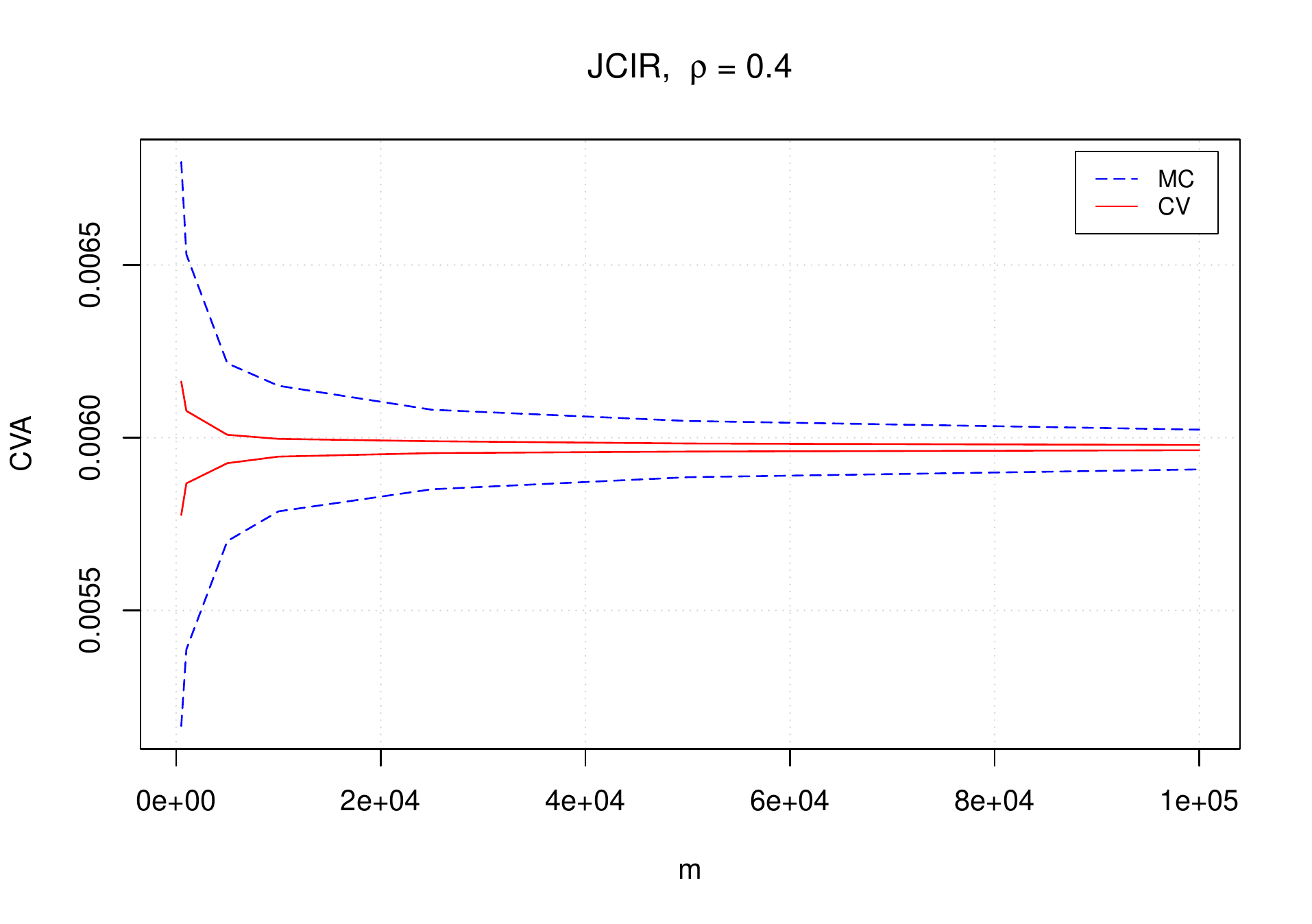}}
\subfigure[CIR$\,(0.02,0.161,0.08,0.03),\,$TC-CIR$\,(0.6,0.512)$]{\includegraphics[width=0.48\columnwidth]{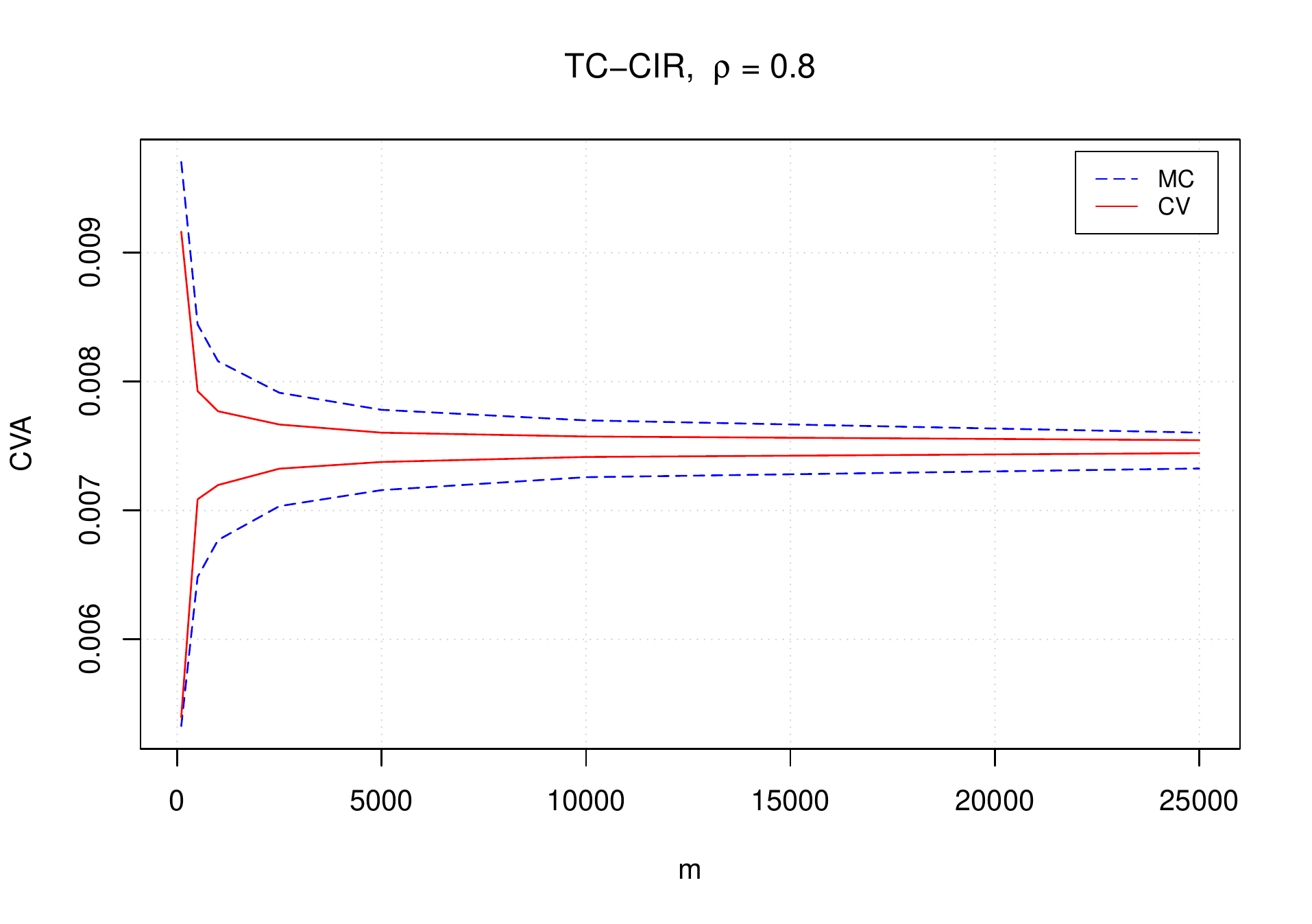}}\hspace{0.25cm}
\subfigure[CIR$\,(0.02,0.161,0.08,0.03),\,$TC-CIR$\,(0.6,0.512)$]{\includegraphics[width=0.48\columnwidth]{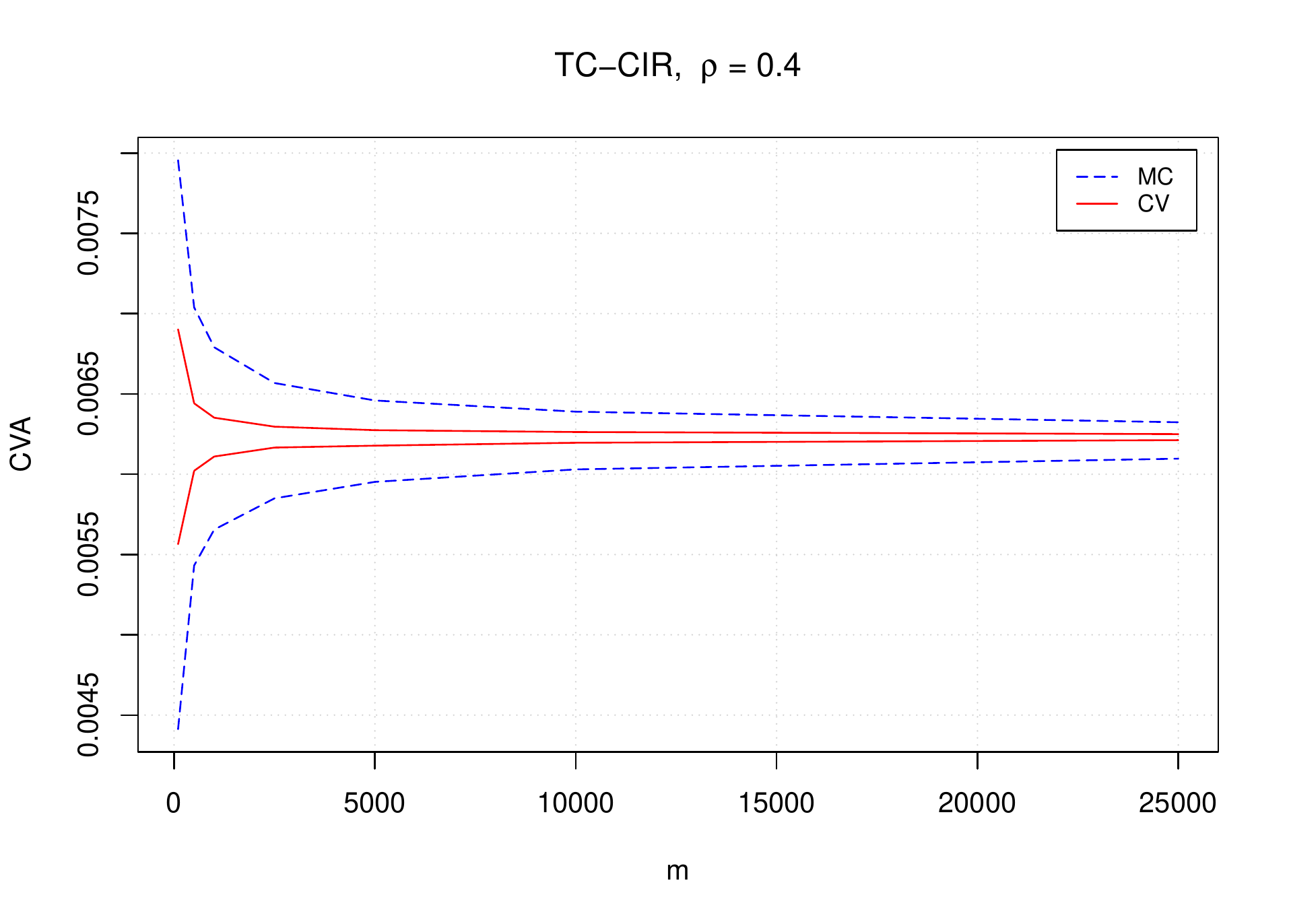}}
\caption{Control variate CVA figures, 3Y Gaussian exposure, $\sigma=8\%$. $h(t)=5\%$.}\label{fig:FigC}
\end{figure}
We observe clearly that the variance reduction technique adopted allows to reduce significantly the computational of the standard Monte Carlo as a solution to CVA computation in presence of WWR. Since $Y$ and and the chosen control $Z$ are more correlated as $|\rho|$ decreases, we observe that the variance is reduced again when $|\rho|$ is decreasing and the convergence of the \textit{adaptive} estimator is faster in this case.
\section{Conclusion}

Among the reduced-form intensity models, affine models like CIR++ process received much attention. The latter consists of a time-homogeneous mean-reverting square-root diffusion shifted in a deterministic way so as to fit a given probability term-structure. In order to increase the attainable volatilities, one can add jumps to the CIR++ dynamics. If the jumps are independent and positive, one obtains the so-called JCIR++ model, which remains affine. The problem however is that the model-implied survival probability curve decreases when increasing the activity of the jumps because they are one-sided. The calibration of the model curve to the market curve being achieved via the shift function, the latter decreases when increasing jumps’ activity. Consequently, in order to avoid facing “negative intensities”, one is limited in the activity of the jumps that can be used. For instance, this specificity limits the attainable values for value-at-risk on CDS, CDS options or wrong-way risk CVA.

An alternative intensity model that allows for two-side jumps is the time-changed idea of Mendoza-Arriaga \& Linetsky. Because jumps can be both positive and negative, the model-implied survival probability curve is expected to decrease less rapidly when increasing the jumps’ activity compared to the JCIR++. Hence, the problem of facing negative shift function (i.e. “negative intensities”) is expected to be less severe, so that larger ``volatility`s effect’’ could be generated. This motivates the use of the above model for CVA purposes. Yet, using the time-changed intensity approach in a multivariate framework requires specific precautions. Without specific adjustments indeed, the time-change technique partly destroys the potential correlation between intensity and exposure which impacts negatively the attainable WWR effect. More importantly, it features forward-looking effects that can generate arbitrage opportunities. In this paper, we have shown how the time-changed model can be used in a consistent and efficient way by reconstructing the exposure dynamics in a ``synchronous'' way such that the above problems can be avoided. The computational issue inherent to the time-changed technique is tackled by proposing a variance reduction technique based on adaptive control variate. Eventually, numerical simulations show that under calibration constraint to a given term structure, the time-change model can give larger WWR effects compared to the CIR++ and JCIR++ models. 

Analyzing the model’s ability to generate higher CDS spread's volatility is another important question from both risk-management and pricing perspectives, which is left for future work.

\section{Acknowledgments}
The research of Cheikh Mbaye is funded by the \textit{National Bank of Belgium} via an \textit{FSR} grant. We would like to thank M. Jeanblanc for discussion on the immersion property when dealing with jumps and time-changed processes. We are also so grateful to G. Pagès for discussions about the adaptive control variate method. Eventually, we thank Laura Ballotta for stimulating discussions about a previous version of this work.

The opinions expressed in this paper are those of the authors and do not necessarily
reflect the views of the \textit{National Bank of Belgium}.

\section{Appendix}
In this section, we give the proof of Corollary~\ref{cor1} and the variance reduction figures using the parameter set of figure~\ref{fig:FigA} (b).
\begin{proof}
Let $V_t=V_s\,e^{-\frac{\sigma^2}{2}(t-s)+\sigma(W^V_t-W^V_s)}$. Clearly, because $V$ is adapted to $\mathbb{F}^{W^V}$,
\begin{equation*}
\mathbb{E}\left[\left.V_t\,\right|\mathcal{F}^{W^V}_s\vee\mathcal{F}^{W^\lambda}_s\right]=\mathbb{E}\left[\left.V_t\,\right|\mathcal{F}^{W^V}_s\right]=V_s\,e^{-\frac{\sigma^2}{2} (t-s)}\mathbb{E}\left[e^{\sigma(W^V_t-W^V_s)}\right]=V_s
\end{equation*}
However, the increments of $W^V$ after $\theta_s>s$ are independent both from $\mathcal{F}^{W^V}_s$ and from $\mathcal{F}^{W^\lambda}_{\theta_s}:=\mathcal{F}^{W^{\lambda^\theta}}_{s}$. Hence,
\begin{eqnarray*}
\mathbb{E}\left[V_t\,|\,\mathcal{F}^{W^V}_s\vee\mathcal{F}^{W^{\lambda^\theta}}_s\right]
&=&V_s\,e^{-\frac{\sigma^2}{2}(t-s)}\mathbb{E}\left[e^{\sigma(W^V_t-W^V_s)} |\mathcal{F}^{W^\lambda}_{\theta_s}\right]\nonumber\\
&=&V_s\,e^{-\frac{\sigma^2}{2}(t-s)}e^{\frac{\sigma^2}{2}(t-\theta_s)}\mathbb{E}\left[e^{\sigma(W^V_{\theta_s}-W^V_s)} |\mathcal{F}^{W^\lambda}_{\theta_s}\right]\nonumber\\
&=&V_s\,e^{-\frac{\sigma^2}{2}(\theta_s-s)}\mathbb{E}\left[e^{\sigma(W^V_{\theta_s}-W^V_s)} |\mathcal{F}^{W^\lambda}_{\theta_s}\right]
\end{eqnarray*}
Assuming in the sequel that $\rho\neq 0$,
\begin{eqnarray*}
\mathbb{E}\left[e^{\frac{\sigma}{\rho}[(W^\lambda_{\theta_s}-W^\lambda_s)-\sqrt{1-\rho^2}(W^\perp_{\theta_s} - W^\perp_s)]} | \mathcal{F}^{W^\lambda}_{\theta_s}\right]
&=&e^{\frac{\sigma}{\rho}(W^\lambda_{\theta_s}-W^\lambda_s)}\mathbb{E}\left[e^{\frac{-\sigma\sqrt{1-\rho^2}}{\rho}(W^\perp_{\theta_s}-W^\perp_s)}| \mathcal{F}^{W^\lambda}_{\theta_s}\right]
\end{eqnarray*}
Observe that $W^\perp_{\theta_s}-W^\perp_s$ is not independent from $\mathcal{F}^{W^\lambda}_{\theta_s}$ as this information set gives us the value of $W^\lambda_{\theta_s}-W^\lambda_s$ :
\[
W^\lambda_{\theta_s}-W^\lambda_s=\rho\left(W^V_{\theta_s}-W^V_s\right)+\sqrt{1-\rho^2}\left(W^\perp_{\theta_s}-W^\perp_s\right)
\]
The computation of the above conditional expectation amounts to evaluate the moment generating function (MGF) $\varphi\left(-\frac{\sigma\sqrt{1-\rho^2}}{\rho}\right)$ associated to the Normal variable $W^\perp_{\theta_s}-W^\perp_s$ for which one knows the value of its weighted sum with another independent variable. More explicitly, we are looking for the MGF of $X=\sqrt{1-\rho^2}\sqrt{\theta_s-s}Z_1$ such that $X+Y= W^\lambda_{\theta_s}-W^\lambda_s$ with $Y=\rho\sqrt{\theta_s-s}Z_2$, $Z_1,Z_2$ iid standard Normal. It can be shown (see eg~\cite{FV17}) that, given $X+\omega_2 Z_2=c$, $X=\omega_1Z_1\sim\mathcal{N}(\tilde{\mu},\tilde{\sigma})$ with  $\tilde{\sigma}^2=\left(\omega_1^{-2}+\omega_2^{-2}\right)^{-1}$ and $\tilde{\mu}=c\tilde{\sigma}^2/\omega_2^2$. Using the values of $\omega_1,\omega_2$ and $c= W^\lambda_{\theta_s}-W^\lambda_s$,  
 \[
\varphi\left(t\right)=e^{(1-\rho^2)( W^\lambda_{\theta_s}-W^\lambda_s)t+\rho^2(1-\rho^2)(\theta_s-s)\frac{t^2}{2}}
\]
Eventually,
\begin{eqnarray*}
\mathbb{E}\left[V_t\,|\,\mathcal{F}^{W^V}_s\vee\mathcal{F}^{W^{\lambda^\theta}}_s\right]
&=&V_s\,e^{-\frac{\sigma^2}{2}(\theta_s-s)}e^{\frac{\sigma}{\rho}(W^\lambda_{\theta_s}-W^\lambda_s)}e^{(1-\rho^2)( W^\lambda_{\theta_s}-W^\lambda_s) \left(-\frac{\sigma\sqrt{1-\rho^2}}{\rho}\right)}\\
&\,&\qquad\times\,e^{\rho^2(1-\rho^2)(\theta_s-s)\left(-\frac{\sigma\sqrt{1-\rho^2}}{\rho\sqrt{2}}\right)^2}\nonumber\\
&=&V_s\,e^{\frac{\sigma}{\rho}\left[1-(1-\rho^2)^{\frac{3}{2}}\right](W^\lambda_{\theta_s}-W^\lambda_s)+\frac{\sigma^2}{2}\left[(1-\rho^2)^2-1\right](\theta_s-s)}\nonumber\\
&\neq &V_s\;.
\end{eqnarray*}
\end{proof}
\newpage
\begin{figure}[H]
\centering
\subfigure[CIR$\,(0.072,0.05,0.045,0.04)$]{\includegraphics[width=0.48\columnwidth]{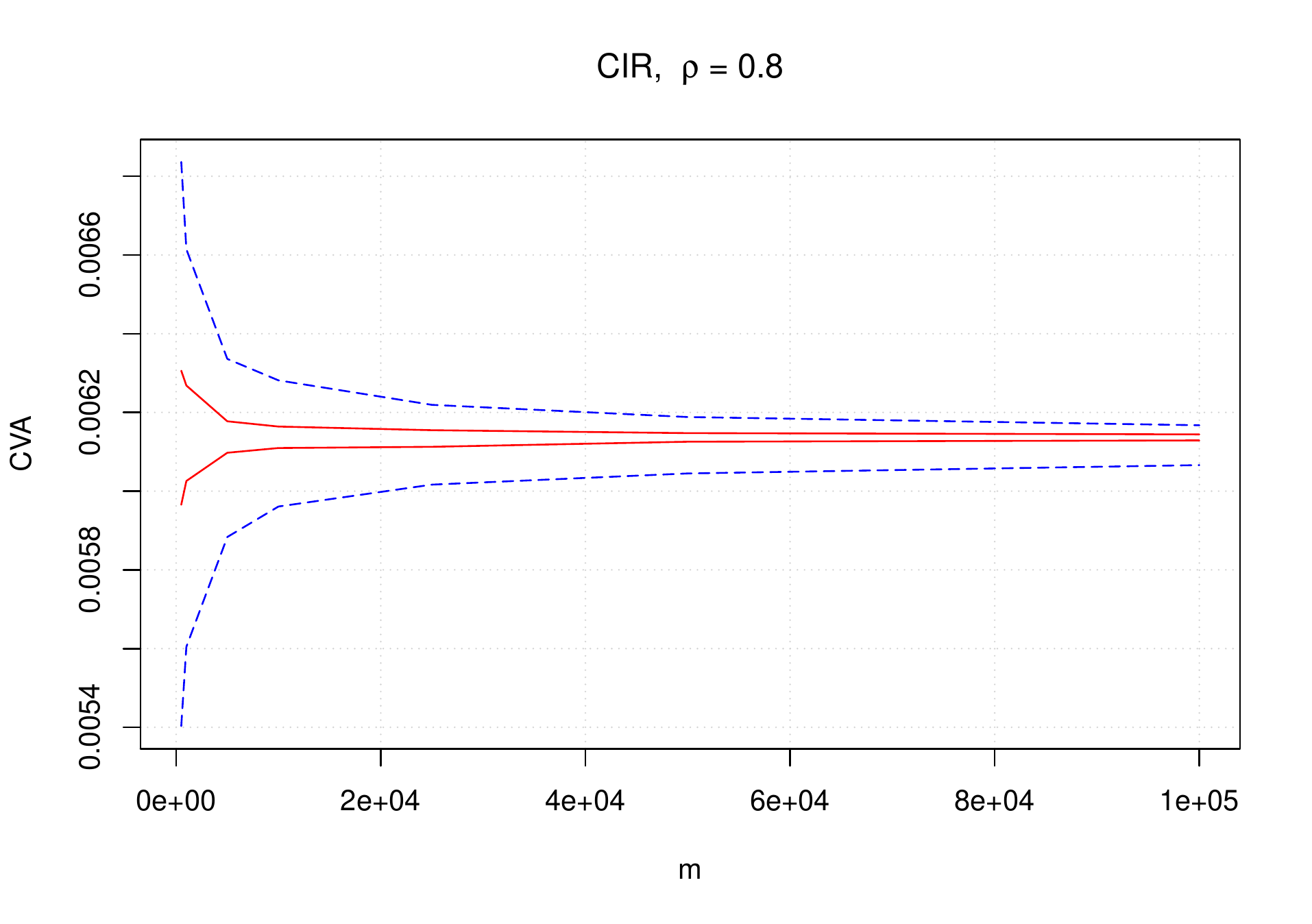}}\hspace{0.25cm}
\subfigure[CIR$\,(0.072,0.05,0.045,0.04)$]{\includegraphics[width=0.48\columnwidth]{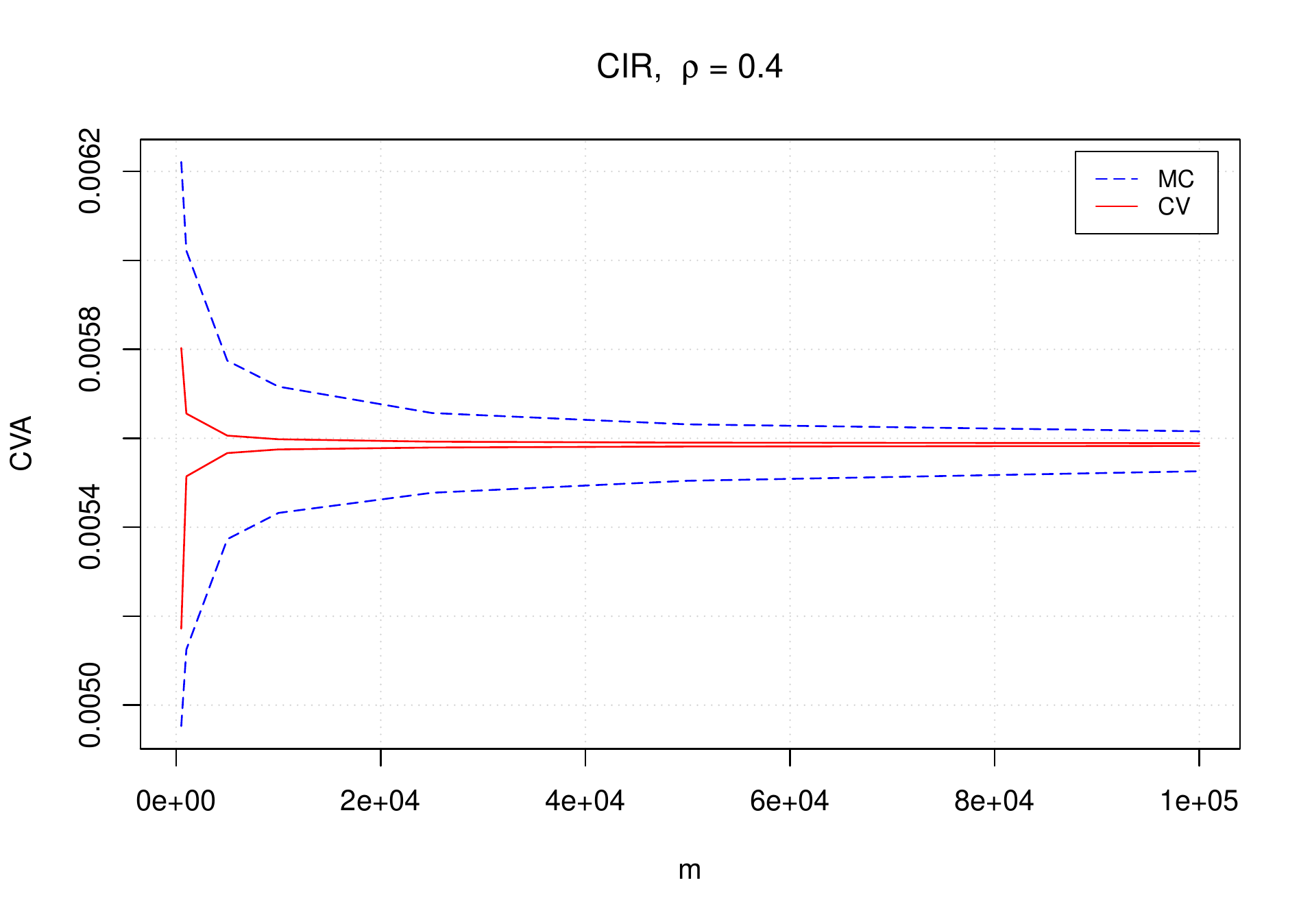}}
\subfigure[CIR$\,(0.072,0.05,0.045,0.04),\,$JCIR$\,(0.07,0.05)$]{\includegraphics[width=0.48\columnwidth]{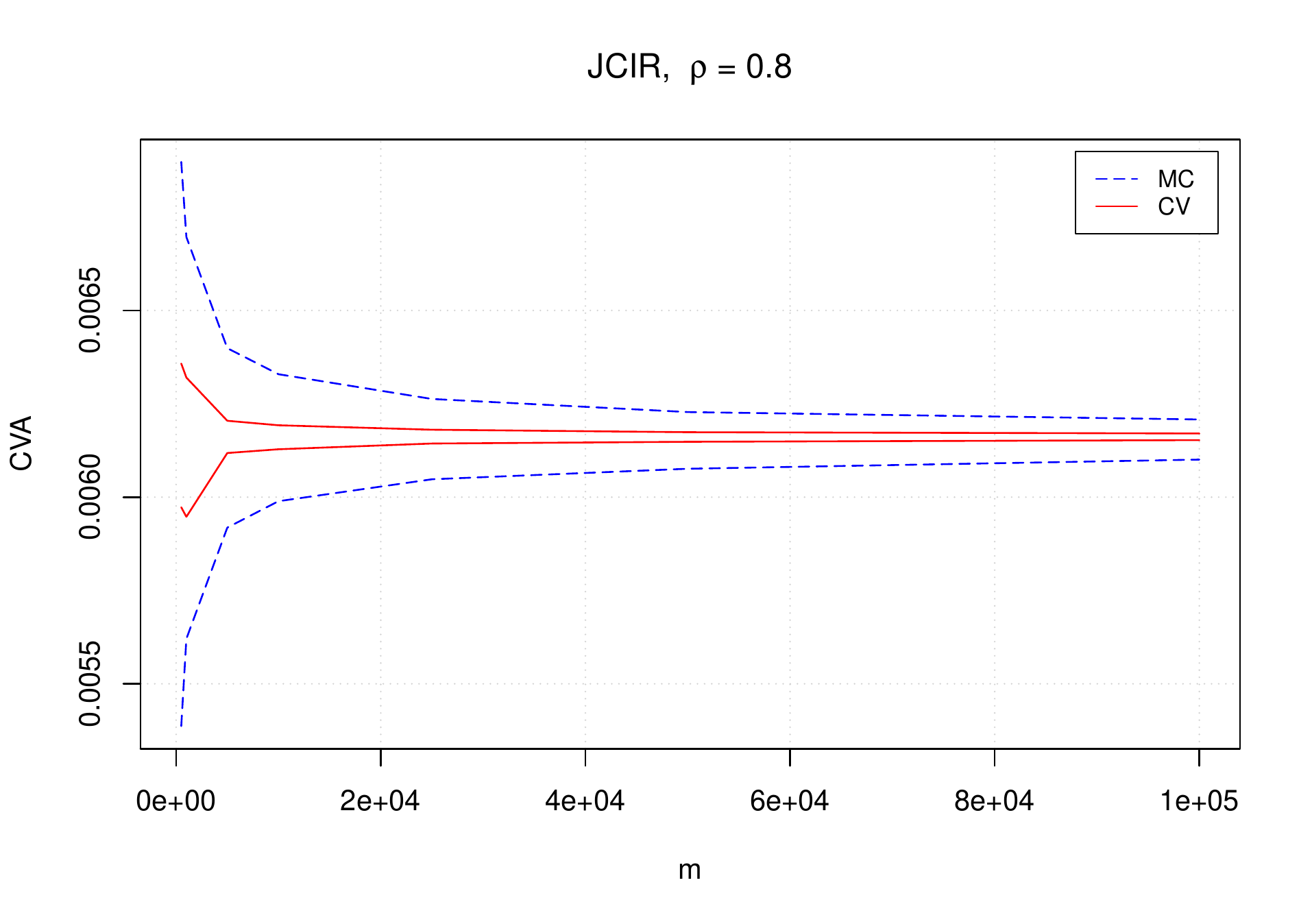}}\hspace{0.25cm}
\subfigure[CIR$\,(0.072,0.05,0.045,0.04)),\,$JCIR$\,(0.07,0.05)$]{\includegraphics[width=0.48\columnwidth]{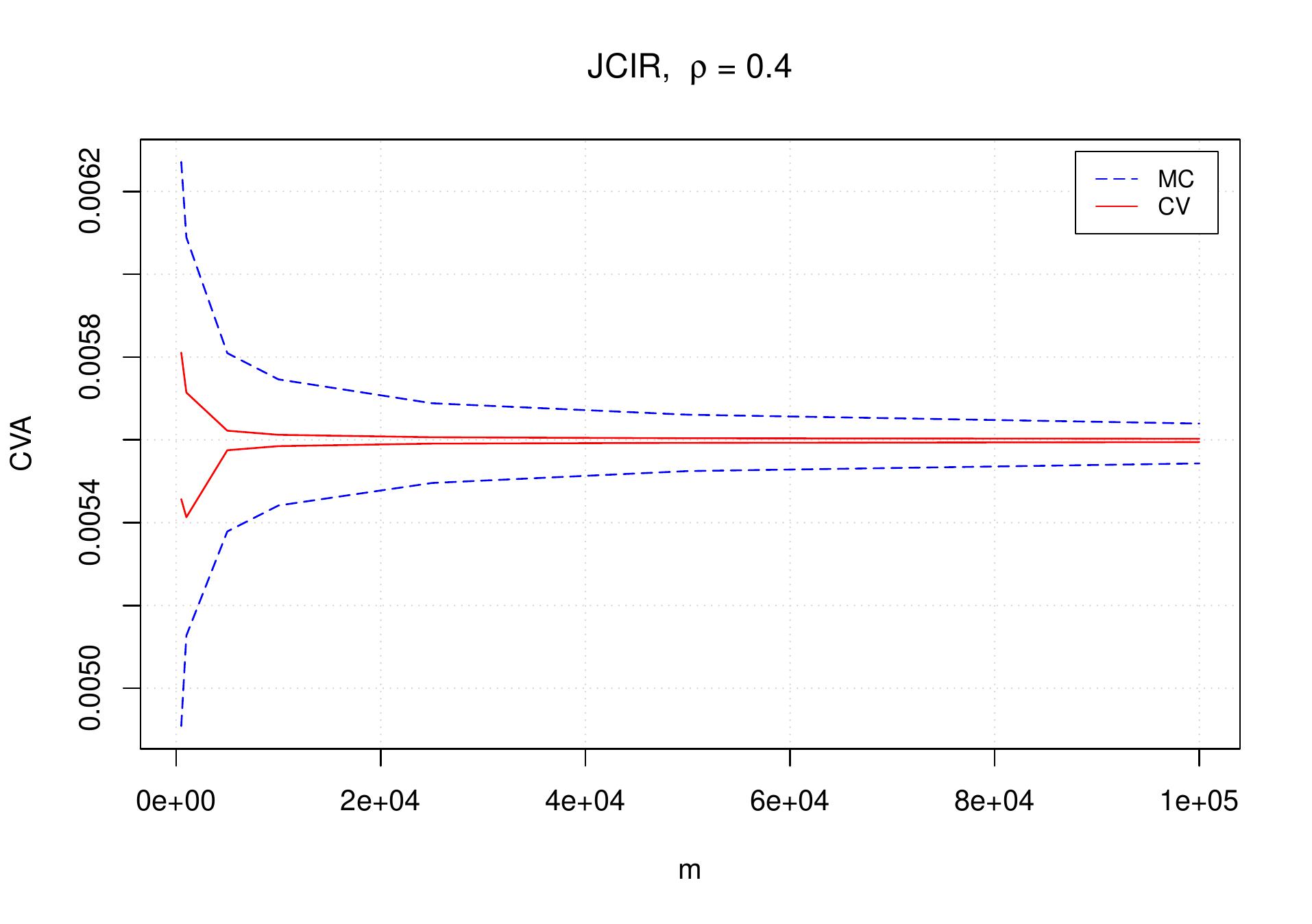}}
\subfigure[CIR$\,(0.072,0.05,0.045,0.04),\,$TC-CIR$\,(0.4,0.49)$]{\includegraphics[width=0.48\columnwidth]{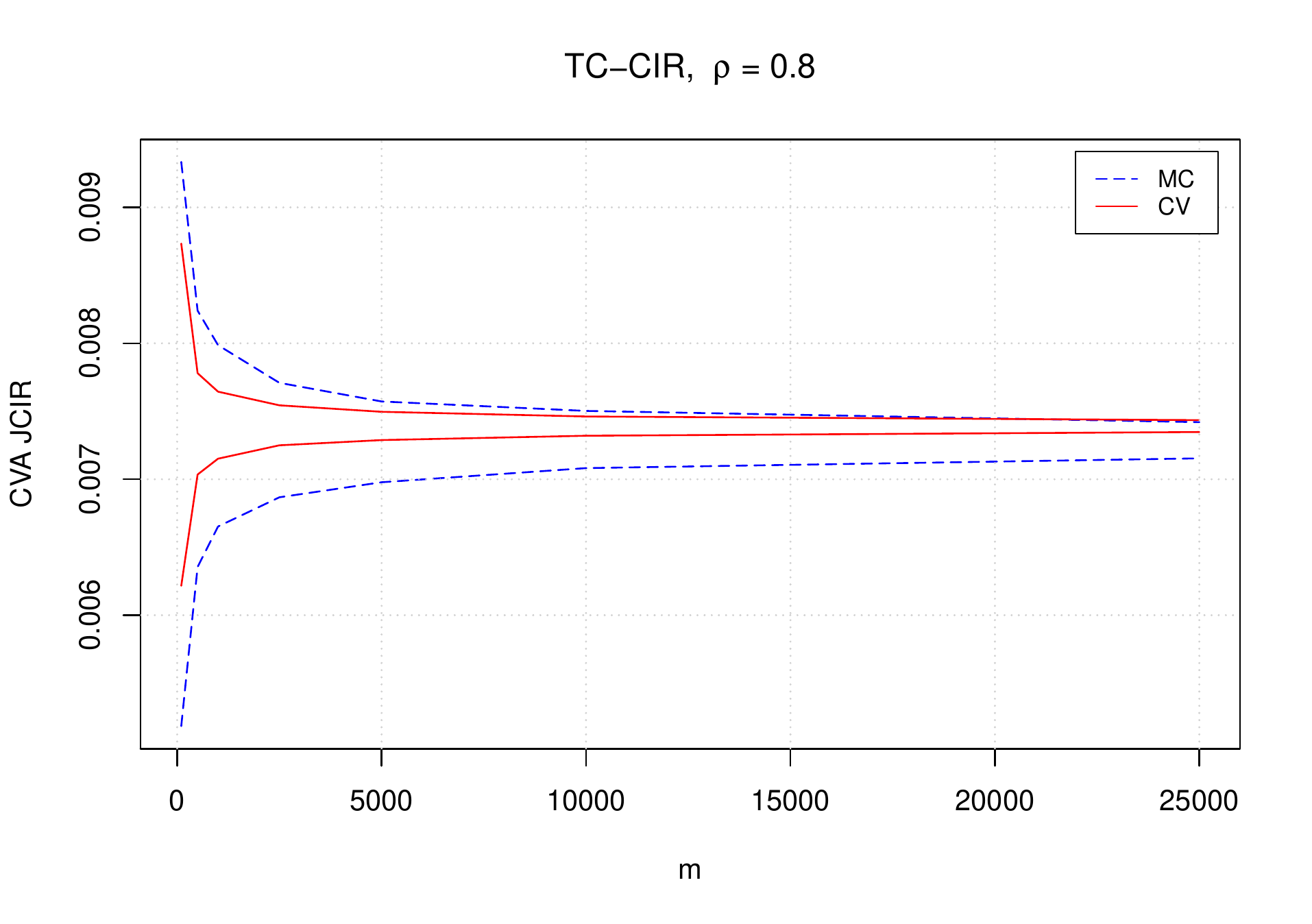}}\hspace{0.25cm}
\subfigure[CIR$\,(0.072,0.05,0.045,0.04),\,$TC-CIR$\,(0.4,0.49)$]{\includegraphics[width=0.48\columnwidth]{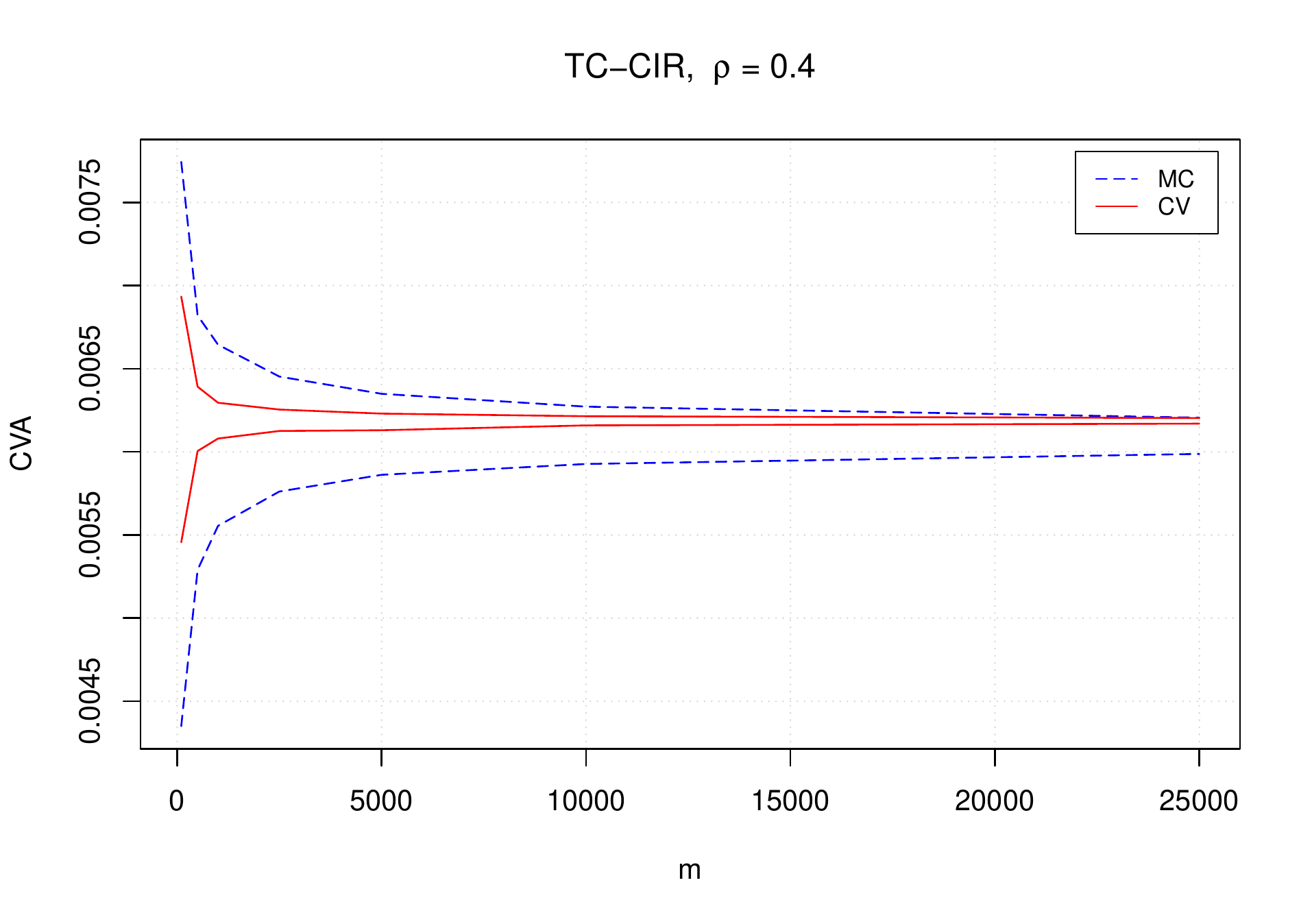}}
\caption{Control variate CVA figures, 3Y Gaussian exposure, $\sigma=8\%$. $h(t)=5\%$.}\label{fig:FigD}
\end{figure}
\newpage
\bibliographystyle{plain}
\bibliography{biblio}
\addcontentsline{toc}{chapter}{References}
\end{document}